\documentclass[12pt]{elsarticle}
\usepackage{mystyleElsArticle}
\usepackage{mathrsfs}
\usepackage{booktabs}

\usepackage[margin=0.99 in]{geometry}    

\renewcommand\d{\mathrm{d}}
\newcommand{\STATE}{\State}

\newcommand{\veps}{\varepsilon}
\def\epsilon{\varepsilon}
\def\bx{\textbf{x}}
\def\bc{\textbf{c}}
\newcommand{\secref}[1]{\S\ref{#1}}

\newtheorem{example}[theorem]{Example}

\makeatletter
\def\ps@pprintTitle{%
	\let\@oddhead\@empty
	\let\@evenhead\@empty
	\let\@oddfoot\@empty
	\let\@evenfoot\@oddfoot
}
\makeatother

\begin{document}
\begin{frontmatter}
\title{Efficient multiscale methods for the semiclassical Schr\"{o}dinger equation with time-dependent potentials}
		
\author[SoochowUniv]{Jingrun Chen}
\ead{jingrunchen@suda.edu.cn}
\author[hku]{Sijing Li}
\ead{lsj17@hku.hk}
\author[hku]{Zhiwen Zhang\corref{cor1}}
\ead{zhangzw@hku.hk}
		
\address[SoochowUniv]{Mathematical Center for Interdisciplinary Research and School of Mathematical Sciences, Soochow University, Suzhou, China.}
\address[hku]{Department of Mathematics, The University of Hong Kong, Pokfulam Road, Hong Kong SAR, China.}
\cortext[cor1]{Corresponding author}		
		
\begin{abstract}
\noindent	
The semiclassical  Schr\"{o}dinger equation with time-dependent potentials is an important model to study electron dynamics under external controls in the mean-field picture. In this paper, we propose two multiscale finite element methods to solve this problem. In the offline stage, for the first approach, the localized multiscale basis functions are constructed using sparse compression of the Hamiltonian operator at the initial time; for the latter, basis functions are further enriched using a greedy algorithm for the sparse compression of the Hamiltonian operator at later times. In the online stage, the Schr\"{o}dinger equation is approximated by these localized multiscale basis in space and is solved by the Crank-Nicolson method in time. These multiscale basis have compact supports in space, leading to the sparsity of stiffness matrix, and thus the computational complexity of these two methods in the online stage is comparable to that of the standard finite element method.
However, the spatial mesh size  in multiscale finite element methods is $ H=\mathcal{O}(\epsilon) $, while $H=\mathcal{O}(\epsilon^{3/2})$ in the standard finite element method, where $\epsilon$ is the semiclassical parameter. By a number of numerical examples in 1D and 2D, for approximately the same number of basis, we show that the approximation error of the multiscale 
finite element method is at least two orders of magnitude smaller than that of the standard finite element method, and the enrichment further
reduces the error by another one order of magnitude.

\medskip

			
\end{abstract}
\end{frontmatter}
	
\section{Introduction} \label{sec:introduction}
\noindent
Precise control of electron dynamics plays a vital role in nanoscale physics. A prototypical example is spintronics in magnetic 
thin films \cite{Zutic:2004}. In the presence of an external current, electron dynamics is driven by the so-called spin-magnetization 
coupling, and magnetization dynamics follows the Landau-Lifshitz equation. Since there is a scale separation between electron dynamics 
and magnetization dynamics in time, a simplification reduces the coupled system into two decoupled equations: electron dynamics is 
driven by magnetization with a prescribed form, and magnetization dynamics is driven by the spin-transfer torque.
Other notable examples include electron dynamics in silicon-based heterojunctions for solar cells \cite{Louwenetal:2016}, 
and light-excited electron dynamics in quantum metamaterials \cite{Quach:2011}.

The objective of this work is to solve a model for electron dynamics in the presence of time-dependent potentials which is often used 
in aforementioned scenarios. To be precise, the underlying Schr\"{o}dinger equation in a dimensionless form reads as
\begin{equation}
\left\{
\begin{aligned}
i\epsilon\partial_t\psi^\epsilon&=-\frac{\epsilon^2}{2}\Delta\psi^\epsilon+v_1^{\epsilon}(\bx) \psi^\epsilon+v_2(\bx, t) \psi^\epsilon,\quad \bx\in D,\quad t\in (t_0, T],\\
\psi^\epsilon &\in H_{\textrm{P}}^{1}(D),\\
\psi^\epsilon|_{t=t_0}&=\psi_{\textrm{in}}(\bx),\quad \bx\in D,
\end{aligned}
\right.
\label{eqn:Sch}
\end{equation}
where $ 0<\epsilon \ll1 $ is a dimensionless constant describing the microscopic and macroscopic scale ratio, $ D = [0,1]^d $ is the spatial domain, $ d $ is the spatial dimension, $[t_0,T]$ is the temporal interval of interest, $\psi^\epsilon = \psi^\epsilon(\bx,t) $ is the wavefunction, and $ \psi_{\textrm{in}}(\bx) $ is the initial data. In \eqref{eqn:Sch} the potential operator consists of two parts: $ v^{\epsilon}_1(\bx) $ contains the microscopic information and $v_2(\bx,t)$ is used to model the external control at the macroscopic scale. Here $ H_{\textrm{P}}^{1}(D) = \{\psi |\psi\in H^1(D) \textrm{ and } \psi \textrm{ is periodic over D} \}$. 

There has been a long history of interest from both mathematical and numerical perspectives to study Schr\"{o}dinger equations; see e.g. \cite{JinActa:2011,Eric:2014} and references therein. In the absence of an external field, $ \psi^\epsilon(\bx, t) $ propagates oscillations with a wavelength of $\mathcal{O}(\epsilon)$. Thus, a uniform $L^2-$approximation of the wavefunction requires the spatial mesh size $h=o(\epsilon)$ and the time step $k=o(\epsilon)$ in the finite element method (FEM) and finite difference method (FDM) \cite{BaoJinMarkowich:2002, JinActa:2011}. If the spectral time-splitting method is employed, a uniform $L^2-$approximation of the wavefunction requires the spatial mesh size $h=\mathcal{O}(\epsilon)$ and the time stepsize $k=o(\epsilon)$ \cite{BaoJinMarkowich:2002}. If $v_1^{\veps}(\bx)$ has some structure, asymptotic methods, such as Bloch decomposition based time-splitting spectral method \cite{Huangetal:2007, Huang:2008}, the Gaussian beam method \cite{Jin:08, Jin:2010, Qian:2010, Yin:2011}, and the frozen Gaussian approximation
method \cite{DelgadilloLuYang:2016}, are proposed and are especially efficient when $\epsilon$ is very small.

With recent developments in nanotechnology, a variety of material devices with tailored functionalities have been fabricated, such as heterojunctions, including the ferromagnet-metal-ferromagnet structure for giant megnetoresistance \cite{Zutic:2004}, the silicon-based heterojunction for solar cells \cite{Louwenetal:2016}, and quantum metamaterials \cite{Quach:2011}. A basic feature of these devices is the combination of dissimilar crystalline structures, which results in a heterogeneous interaction from ionic cores with different lattice structures. Therefore, when traveling through a device, electrons experience a potential $v_1^{\epsilon}(\bx)$ which is typically discontinuous and has no separation of scales. Consequently, all the available methods based on asymptotic analysis cannot be applied. Moreover, direct methods, such as FEM and FDM, are extremely inefficient with strong mesh size restrictions. This motivates us to design efficient numerical methods for \eqref{eqn:Sch} in the general situation.

Our recent work \cite{chen2019multiscale} can solve \eqref{eqn:Sch} with a generic $v_1^{\veps}(\bx)$ in the absence of the time-dependent potential $v_2(\bx,t)$,
which is motivated by the multiscale finite element method (MsFEM) for solving elliptic problems with multiscale coefficients \cite{HouWuCai:99,EfendievHou:09, Peterseim:2014, Owhadi:2015,Owhadi:2017, hou2017sparse}. MsFEM is capable of correctly capturing the large scale components of the multiscale solution on a coarse grid without accurately resolving all the small scale features in the solution. This is accomplished by incorporating the local microstructures of the differential operator into multiscale basis functions. 

Inspired by \cite{chen2019multiscale}, we will develop two MsFEMs to solve Schr\"{o}dinger equation with a generic $v_1^{\veps}(\bx)$ in the presence of the time-dependent potential $v_2(\bx,t)$. The main ingredient of the proposed methods is the construction of multiscale basis functions with time-dependent information. In the first method, the localized multiscale basis functions are constructed using sparse compression of the Hamiltonian operator at the initial time; in the second method, the enriched multiscale basis functions are added using sparse compression of the Hamiltonian operator at latter times. In both methods, $ H=\mathcal{O}(\epsilon) $, while a stronger mesh condition is required in the standard FEM. Numerical examples in 1D with a periodic potential, a multiplicative two-scale potential, and a layered potential, and in 2D with a checkboard potential are tested to demonstrate the robustness and accuracy of the proposed methods. 

For time-dependent potentials, it is worth mentioning that effective methods have been developed for the temporal approximation; see e.g. \cite{iserles2018magnus, iserles2019compact, iserles2019solving}. It will be of great interest to study how the temporal approximation approach and our method can be combined since the wavefunction oscillates in both spatial and temporal directions. We shall investigate this issue in our future work.

The rest of the paper is organized as follows. In \secref{sec:MsFEM}, we introduce the MsFEM and  enriched MsFEM (En-MsFEM) for the semiclassical Schr\"{o}dinger equations with time-dependent and multiscale potentials and discuss their properties. Numerous numerical results are presented in \secref{sec:NumericalExamples}, including both one dimensional and two dimensional examples to demonstrate the robustness and accuracy of the proposed methods. Conclusions are drawn in \secref{sec:Conclusion}.
	
\section{A Multiscale finite element method for Schr\"{o}dinger equation} \label{sec:MsFEM}
\noindent
The construction of multiscale basis functions for time-dependent and multiscale potentials is mainly based on the approach in \cite{chen2019multiscale} for time-independent potentials.
 
\subsection{Construction of multiscale basis functions} \label{sec:OC}
\noindent
Define the Hamiltonian operator $\mathcal{H}(t)(\cdot)\equiv -\frac{\epsilon^2}{2}\Delta(\cdot)+v_1^{\epsilon}(\bx)(\cdot) + v_2(\bx,t)(\cdot)$ 
and introduce the following energy notation $ ||\cdot||_{V(t)} $ for Hamiltonian operator
\begin{align}\label{eqn:energynorm}
||\psi^\epsilon||_{V(t)}=\frac12(\mathcal{H}\psi^\epsilon,\psi^\epsilon)=\frac12\int_{D}\Big(\frac{\epsilon^2}{2}|\nabla \psi^\epsilon|^2+v_1^{\epsilon}(\bx)|\psi^\epsilon|^2+v_2(\bx,t)|\psi^\epsilon|^2\Big) \d\bx.
\end{align}
	
Note that \eqref{eqn:energynorm} does not define a norm since $v_1^{\epsilon}$ and $v_2$ usually can be negative, and thus the bilinear form
associated to this notation is not coercive, which is quite different from the case of elliptic equations. However, this does not
mean that available approaches \cite{HouWu:97, BabuskaLipton:2011, Peterseim:2014, Owhadi:2017, hou2017sparse} cannot be applied for
the Schr\"{o}dinger equation. In fact, we shall utilize the similar idea to construct localized multiscale finite element basis
functions on a coarse mesh by an optimization approach using the above energy notation $ ||\cdot||_{V(t)} $ for the Hamiltonian operator.

To construct such localized basis functions,  we first partition the physical domain $D$ into a set of regular coarse elements with mesh size $H$. For example, we divide $D$ into a set of non-overlapping triangles $\mathcal{T}_{H}$, such that no vertex of one triangle lies in the interior of the edge of another triangle. In each element $K\in \mathcal{T}_{H}$, we define a set of nodal basis $\{\varphi_{j,K},j=1,...,k\}$ with $k$ being the number of nodes of the element $K$. From now on, we neglect the subscript $K$ for notational convenience. The functions $\varphi_{i}(\bx)$ are called measurement functions, which are chosen as the characteristic functions on each coarse element in \cite{hou2017sparse,Owhadi:2017} and piecewise linear basis functions in \cite{Peterseim:2014}.
	
Let $\mathcal{N}$ denote the set of vertices of $\mathcal{T}_{H}$ (removing the repeated vertices due to the periodic boundary condition) and $N_{H}$ be the number of vertices. For every vertex $\bx_i\in\mathcal{N}$, let $\varphi_{i}^{H}(\bx)$ denote the corresponding FEM nodal basis function, i.e., $\varphi_{i}^{H}(\bx_j)=\delta_{ij}$. Then, we can solve optimization problems to obtain the multiscale basis functions. Specifically, let $\phi_i(\bx)$ be the minimizer of the following constrained optimization problem
\begin{align}
	\phi_i& =\underset{\phi \in H_{\textrm{P}}^{1}(D)}{\arg\min} ||\phi||_{V(t)}   \label{OC_SchGLBBasis_Obj}\\
	\text{s.t.}\ &\int_{D}\phi \varphi_{j}^{H} \d\bx= \delta_{i,j},\ \forall 1\leq j \leq N_{H}.  \label{OC_SchGLBBasis_Cons1}
\end{align}
The superscript $\epsilon$ is dropped for notational simplicity and the periodic boundary condition is incorporated into the above optimization problem through the solution space $H_{\textrm{P}}^{1}(D)$. 
The minimizers of \eqref{OC_SchGLBBasis_Obj} - \eqref{OC_SchGLBBasis_Cons1}, i.e.,  $\phi_i$, $i=1,...,N_H$ will be referred as the multiscale basis functions. Let $V^{H}$ denote the space spanned by the  multiscale basis functions $\phi_i$. Namely $V^{H}=\{\phi_{i}(\bx):i=1,...,N_H \}$. From the construction process, we know that $V^{H}\subset H_{\textrm{P}}^{1}(D)$.

In general, one cannot solve the above optimization problem analytically. Therefore, we use numerical methods to solve it. Specifically, we partition the physical domain $D$ into a set of non-overlapping fine triangles with size $h \ll \veps$. Let $\varphi_{s}^{h}(\bx)$, $s=1,...,N_h$ denote the fine-scale FEM nodal basis with mesh size $h$, where $N_h$ is the total number of the nodal basis, and let  $V^{h}=\{\varphi_{s}^{h}(\bx)\}_{s=1}^{N_h}$ denote the FEM space. 
Then, we use standard FEM basis to represent $\phi_i(\bx)$, $\varphi_{j}(\bx)$, $1\leq i,j\leq N_{H}$. In the discrete level, the optimization problem \eqref{OC_SchGLBBasis_Obj} - \eqref{OC_SchGLBBasis_Cons1} is reduced to a constrained quadratic optimization problem, which can be efficiently solved using Lagrange multiplier methods. Finally, with these multiscale basis functions $ \{\phi_i(\bx)\}_{i=1}^{N_H}$, we can solve the Schr\"{o}dinger equation \eqref{eqn:Sch} using the Galerkin method.
\begin{remark}
In analogy to MsFEM \cite{HouWuCai:99,EfendievHou:09},  the multiscale basis functions $\{\phi_i(\bx)\}_{i=1}^{N_H}$ are defined on coarse elements with mesh size $H$. However, they are represented by fine-scale FEM basis with mesh size $h$, which can be pre-computed in parallel. 
\end{remark}	
\begin{remark}
Note that the energy notation $||\cdot||_{V(t)}$ in \eqref{eqn:energynorm} does not define a norm. However, as long as $v_1^{\veps}(\bx)$ and $v_2(\bx, t)$ are bounded from below and the fine mesh size $h$ is small enough, the discrete problem of \eqref{OC_SchGLBBasis_Obj} - \eqref{OC_SchGLBBasis_Cons1} is convex and thus admits a unique solution; see \cite{hou2017sparse, LiZhangCiCP:18} for details.
\end{remark}

We shall show that the multiscale basis functions $\{\phi_i(\bx)\}_{i=1}^{N_H}$ decay exponentially fast away from its associated vertex $x_i\in\mathcal{N}$ under certain conditions. This allows us to localize the basis functions to a relatively smaller domain and reduce the computational cost.
	
In order to obtain localized basis functions, we first define a series of nodal patches $\{D_{\ell}\}$ associated with $\bx_i\in\mathcal{N}$ as
\begin{align}
	D_{0}&:=\textrm{supp} \{ \varphi_{i}^{H} \} = \cup\{K\in\mathcal{T}_H | \bx_i \in K \}, \label{nodal_patch0}\\
	D_{\ell}&:=\cup\{K\in\mathcal{T}_H | K\cap \overline{D_{\ell-1}} \neq \emptyset\}, \quad  \ell=1,2,\cdots.
	\label{nodal_patchl}
\end{align}	
\begin{assumption} \label{CoarseMeshResolution}
	We assume that the potential term $v_1^{\epsilon}(\bx)+v_2(\bx,t)$ is uniformly bounded, i.e.,  $V_0 := ||v_1^{\epsilon}(\bx)+v_2(\bx,t)||_{L^\infty (D;[t_0,T])}< +\infty$ and the   mesh size $H$ of $\mathcal{T}_{H}$ satisfies
	\begin{equation}
	\label{eqn:meshcondition}
	\sqrt{V_0}H/\epsilon\lesssim 1,
	\end{equation}
	where $\lesssim$ means bounded from above by a constant.
\end{assumption}
\noindent Under this resolution assumption for the coarse mesh, many typical potentials in the Schr\"{o}dinger equation \eqref{eqn:Sch} can be treated as a perturbation to the kinetic operator. Thus, they can be computed using our method. Then, we can show that the multiscale finite element basis functions have the exponentially decaying property.
\begin{proposition}[Exponentially decaying property]\label{ExponentialDecay}
	Under the resolution condition of the coarse mesh, i.e., \eqref{eqn:meshcondition}, there exist constants $C>0$ and $0<\beta<1$ independent of $H$, such that
	\begin{equation}\label{eqn:exponentialdecay}
	||\nabla \phi_i(\bx) ||_{L^2(D\backslash D_{\ell})}\leq C \beta^{\ell} ||\nabla \phi_i(\bx) ||_{L^2(D)},
	\end{equation}
	for any $i=1, 2, ..., N_H$.
\end{proposition}

Proof of \eqref{eqn:exponentialdecay} will be given in \cite{ChenMaZhang:prep}. The main idea is to combine an iterative Caccioppoli-type argument \cite{Peterseim:2014, LiZhangCiCP:18} and some refined estimates with respect to $\veps$. 
		
The exponential decay of the basis functions enables us to localize the support sets of
the basis functions $\{\phi_i(\bx)\}_{i=1}^{N_H}$, so that the corresponding stiffness matrix is sparse and the
computational cost is reduced. In practice, we define a modified constrained optimization problem as follows
\begin{align}
	\phi_i^{\textrm{loc}}& =\underset{\phi \in H_{\textrm{P}}^{1}(D)}{\arg\min} ||\phi||_{V(t)}   \label{OC_SchBasis_Obj}\\
	\text{s.t.}\ &\int_{D_{l^*}}\phi \varphi_{j}^{H} \d\bx= \delta_{i,j},\ \forall 1\leq j \leq N_{H}, \label{OC_SchBasis_Cons1}\\
	&\phi(\bx)= 0, \ x\in D\backslash D_{l^*},\label{OC_SchBasis_Cons2}
\end{align}
where $D_{l^*}$ is the support set of the localized multiscale basis function $\phi_i^{\textrm{loc}}(\bx)$ and the choice of the integer $l^*$ depends on the decaying speed of $\phi_i^{\textrm{loc}}(\bx)$. In \eqref{OC_SchBasis_Cons1} and \eqref{OC_SchBasis_Cons2}, we have used the fact that $\phi_i(\bx)$ has the exponentially decaying property so that we can localize the support set of $\phi_i(\bx)$ to a smaller domain $D_{l^*}$. In numerical experiments, we find that a small integer $l^*\sim \log(L/H)$  will give accurate results, where $L$ is the diameter of domain $D$. Moreover, the optimization problem \eqref{OC_SchBasis_Obj} - \eqref{OC_SchBasis_Cons2} can be solved in parallel. Therefore, the exponentially decaying property significantly reduces our computational cost in constructing basis functions and computing the solution of the Schr\"{o}dinger equation \eqref{eqn:Sch}.

\subsection{Spatial and temporal discretization} \label{sec:NumCom}
\noindent	
Given the set of multiscale basis functions $ \{\phi_i(\bx)\}_{i=1}^{N_H}$ obtained in \eqref{OC_SchBasis_Obj} - \eqref{OC_SchBasis_Cons2} at the initial time $t=t_0$ (superscripts dropped for convenience), we can approximate the wave function by $ \psi^\epsilon(\bx,t)=\sum_{i=1}^{N_H}c_i(t)\phi_i(\bx) $ using the Galerkin method. Therefore, the coefficients $ c_i(t),i=1,...,N_H $ satisfies a system of ordinary differential equations
\begin{align*}
\left(i\epsilon\partial_t\sum_{i=1}^{N_H}c_i(t)\phi_i(\bx),\phi_j(\bx)\right)=\left(\mathcal{H}(t)\sum_{i=1}^{N_H}c_i(t)\phi_i(\bx),\phi_j(\bx)\right),~\bx\in D,~t\in (t_0, T],~j=1,\cdots,N_H,
\end{align*}
which can be rewritten in a semi-discrete form
\begin{align*}
i\epsilon M\frac{d\bc}{dt}=\big(\frac{\epsilon^2}{2}S+V_1+V_2(t)\big)\bc,
\end{align*}
where $ \bc=(c_1(t),c_2(t),...,c_{N_H}(t))^T $, $ \frac{d\bc}{dt}=(\frac{dc_1(t)}{dt},\frac{dc_2(t)}{dt},...,\frac{dc_{N_H}(t)}{dt})^T $,
and $S$, $M$, $V_1$, and $V_2(t)$ are matrices with dimension $N_H \times N_H$ with their entries given by 
\begin{align*}
S_{i,j}&=\int_{D}\nabla \phi_i \cdot \nabla \phi_j \d\bx, \quad M_{i,j}=\int_{D} \phi_i \phi_j \d\bx, \\
(V_1)_{i,j}&=\int_{D}\phi_i v_1^{\epsilon}(\bx) \phi_j \d\bx, \quad (V_2(t))_{i,j}=\int_{D}\phi_i v_2(\bx,t) \phi_j \d\bx.
\end{align*}

For the temporal direction, we apply the Crank-Nicolson method. Let $\bc^k$ be the numerical approximation of $\bc(t_k)$ at time $t_k= t_0 + k\Delta t$ with 
$\Delta t$ the temporal stepsize and $k=0,1,\cdots$. The fully discrete form is 
\begin{align}\label{eqn:fulldiscretescheme0}
i\epsilon M\frac{\bc^{k+1}-\bc^k}{\Delta t}=(\frac{\epsilon^2}{2}S+V_1+V_2(t_{k+\frac{1}{2}}))\frac{\bc^{k+1}+\bc^k}{2},
\end{align}
where $t_{k+\frac{1}{2}}=\frac{t_{k+1}+t_k}{2}$, or 
\begin{align}\label{eqn:fulldiscretescheme}
\left\{i\epsilon M - \frac{\Delta t}{2}\Big(\frac{\epsilon^2}{2}S+V_1+V_2(t_{k+\frac{1}{2}})\Big)\right\}\bc^{k+1}=\left\{i\epsilon M + \frac{\Delta t}{2}\Big(\frac{\epsilon^2}{2}S+V_1+V_2(t_{k+\frac{1}{2}})\Big)\right\}\bc^k,
\end{align}
equivalently.

By solving \eqref{eqn:fulldiscretescheme}, we obtain $\bc^{k+1}$ and the approximate wavefunction at $t_{k+1}$ is
\begin{align}\label{eqn:psiapprox}
\psi_{k+1} = \sum_{i=1}^{N_H}c_i^{k+1}\phi_i(\bx).
\end{align}

\begin{remark}
If $v_2(\bx,t)$ has an affine form, i.e., $ v_2(\bx,t) = \sum_{n=1}^{r} v_{2,n}(\bx)s_{n}(t) $,  using separation of variables, we compute $(V_{2,n})_{i,j}=\int_{D}\phi_i v_{2,n}(\bx) \phi_j \d\bx$, $i,j=1,...,N_H$, $n=1,...,r$ and save them in the offline stage. 
This leads to a considerable saving in assembling the matrix for $V_2(t)$ at different times. 		
\end{remark}
 	
\subsection{An enriched multiscale finite element method for Schr\"{o}dinger equations}\label{sec:enriched-MsFEM} 	
\noindent 
Time-dependent potentials may vary dramatically in time, which will be adopted in the construction of multiscale
basis functions. The En-MsFEM consists of an initial construction stage and an enrichment stage. In the initial construction stage, we solve \eqref{OC_SchBasis_Obj} - \eqref{OC_SchBasis_Cons2} 
at the initial time $t=t_0$ and obtain multiscale basis functions $\phi_i(\bx)$, $i=1,...,N_H$ and 
$V^{H}=\{\phi_i(\bx): i=1,...,N_H\}$. $V^{H}$ only contains the information of $v_2(\bx, t_0)$, which may have the
limited approximation accuracy when $v_2(\bx,t)$ has large changes in time.

In the enrichment stage, we add extra multiscale basis functions into $V^{H}$ by taking into account $v_2(\bx,t)$ at later times. Precisely, we choose a set of time instances as $t_0 < t_1 < \cdot\cdot\cdot < t_{N_t}=T$ and generate the corresponding snapshots of $v_2(\bx,t_{\ell})$, $0\le \ell\le N_t$. A brute-force strategy is to enrich $V^H$ at every time 
step $\ell$, $0\le \ell\le N_t$, by solving \eqref{OC_SchBasis_Obj} - \eqref{OC_SchBasis_Cons2} at $t_{\ell}$. 
This strategy is very expensive since $\Delta t$ has to be $\veps$ dependent due to the $\mathcal{O}(\veps)$ oscillations in time,
thus the dimension of $V^H$ grows dramatically. However, there is a continuous dependence of minimizers 
to \eqref{OC_SchBasis_Obj} - \eqref{OC_SchBasis_Cons2} on the potential function and the temporal variation
of the potential does not have $\mathcal{O}(\veps)$ dependence.
 
Therefore, we propose a greedy algorithm in the enrichment stage. The following result states the continuous dependence of multiscale basis functions on the potential function, whose proof is given in \ref{PhiDependOnV}.
\begin{theorem}\label{basis-dependon-potential}
Given two time instances $t_{\ell_1}$ and $t_{\ell_2}$, and mesh size of the fine-scale triangles is small such that: 
(1) $h/\veps=\kappa $ is small; and (2) $h^d  \lVert v_2(\cdot,t_{\ell_1}) - v_2(\cdot,t_{\ell_2}) \rVert_{L^{\infty}(D)}<1$, 
then the corresponding unique minimizers of \eqref{OC_SchBasis_Obj} - \eqref{OC_SchBasis_Cons2} satisfy
\begin{align}
\lVert \phi(\cdot,t_{\ell_1}) - \phi(\cdot,t_{\ell_2}) \rVert_{L^{\infty}(D)} \leq \frac{C}{\kappa^6} \veps^{-2} \lVert v_2(\cdot,t_{\ell_1}) - v_2(\cdot,t_{\ell_2}) \rVert_{L^{\infty}(D)},
\end{align}
where the constant $C$ is independent of $h, \veps$, and $\lVert v_2(\cdot,t_{\ell_1}) - v_2(\cdot,t_{\ell_2}) \rVert_{L^{\infty}(D)}$.
\end{theorem}

Now we are in the position to introduce the greedy algorithm. Firstly, we choose a time instance $t_{{\ell}_1}$, 
 so that the quantity $\lVert v_2(\cdot,t_{\ell}) \rVert_{L^{\infty}(D)} $ is maximized over $0 < \ell \le N_t$.
Solving \eqref{OC_SchBasis_Obj} - \eqref{OC_SchBasis_Cons2} at $t=t_{\ell_1}$ generates another set of 
multiscale basis functions, denoted by $\phi_i(\bx;t_{\ell_1})$, $i=1,...,N_H$.
Then, we search over the remaining time instances and find a time instance $t_{\ell_2}$, so that the quantity  $\lVert v_2(\cdot,t_{\ell_2})-v_2(\cdot,t_{\ell_1}) \rVert_{L^{\infty}(D)}$ is maximized among all remaining time instances.  
Solving \eqref{OC_SchBasis_Obj} - \eqref{OC_SchBasis_Cons2} at $t=t_{\ell_2}$ generates another set of 
multiscale basis functions, denoted by $\phi_i(\bx;t_{\ell_2})$, $i=1,...,N_H$.
This procedure is repeated until the quantity  $\lVert v_2(\cdot,t_{s})-v_2(\cdot,t_{r}) \rVert_{L^{\infty}(D)}$ is smaller than a given threshold $\delta$, where $t_{s}$ represents any time instance selected and $t_{r}$ represents any time instance left.
Finally, all multiscale basis functions generated earlier form the enriched multiscale finite element space $V^{E}$, which will be
used as the approximation space in the Galerkin method. 

Note that $V^H\subset V^E$, thus a better approximation is always
expected for $V^E$, as verified in Section \ref{sec:NumericalExamples}.
Practically, a post-processing on $V^E$, such as Gram-Schmidt orthogonalization, may be needed to get rid of the nearly dependent basis 
and reduce the condition number of the stiffness matrix. 
Since in most real applications, the potential function $v_2(\bx,t)$ is periodic in $t$, only time instances within one period are taken into account. 
Below is the complete algorithm to enrich multiscale basis functions.
\begin{algorithm}[tbph]
\caption{A greedy algorithm to enrich the set of multiscale basis functions}\label{alg:GreedyAlgorithm}
\begin{algorithmic}[1]
\STATE Set up time instances as $t_0 < t_1 < \cdot\cdot\cdot < t_{N_t}=T$ and a threshold $\delta$; 
let $S=[~]$ be the set for selected time instances, $R=\{t_0,t_1,...,t_{N_t}\}$ be the set for remaining time instances,
and $V^{E}=[~]$ be the set of multiscale basis functions. 
\STATE Solve the optimization problem \eqref{OC_SchBasis_Obj} - \eqref{OC_SchBasis_Cons2} with the potential $v_1^{\epsilon}(\bx)+v_2(\bx,t_0)$ to obtain multiscale basis functions $\phi_i(\bx)$, $i=1,...,N_H$ and $V^{E}=\{\phi_i(\bx),i=1,...,N_H\}$. 
\STATE Find $t_{\ell_1}$ so that $\lVert v_2(\cdot,t_{\ell_1}) \rVert_{L^{\infty}(D)} $ is maximized. 
Solve the optimization problem \eqref{OC_SchBasis_Obj} - \eqref{OC_SchBasis_Cons2} with the potential $v_1^{\epsilon}(\bx)+v_2(\bx,t_{\ell_1})$
to obtain multiscale basis functions $\phi_i(\bx;t_{\ell_1})$, $i=1,...,N_H$; set $S=[t_{\ell_1}]$, $R=R\setminus\{t_{\ell_1}\}$, and
$V^E = V^E\bigcup\{\phi_i(\bx;t_{\ell_1}),i=1,...,N_H\}$.
\While {$\lVert v_2(\cdot,t_{s})-v_2(\cdot,t_{r}) \rVert_{L^{\infty}(D)}>\delta$, 
	where $t_{s}\in S$ and $t_{r}\in R$.}  
\STATE  Find $t_{r*}\in R$ so that $\lVert v_2(\cdot,t_{s})- v_2(\cdot,t_{r*})\rVert_{L^{\infty}(D)} $ is maximized, where $t_{s}\in S$ and $t_{r*}\in R$;
\STATE  Solve the optimization problem \eqref{OC_SchBasis_Obj} - \eqref{OC_SchBasis_Cons2} with the potential $v_1^{\epsilon}(\bx)+v_2(\bx,t_{r*})$
to obtain multiscale basis functions $\phi_i(\bx;t_{r*})$, $i=1,...,N_H$; 
\STATE  Set $S=S\cup\{t_{r*}\}$, $R=R\setminus\{t_{r*}\}$, and $V^E = V^E\bigcup\{\phi_i(\bx;t_{r^*}),i=1,...,N_H\}$.
\EndWhile
\STATE Post-process on $V^{E}$. 
\end{algorithmic}
\end{algorithm}

\textbf{Algorithm \ref{alg:GreedyAlgorithm}} is very efficient in the sense that only one-step enrichment, i.e., steps 1 - 3 in the greedy algorithm, 
is enough to capture the time-dependent feature of the wavefunction; see numerical results in Section \ref{sec:NumericalExamples} for details.
The underlying reason is that the second assumption $h^d  \lVert v_2(\cdot,t_{\ell_1}) - v_2(\cdot,t_{\ell_2}) \rVert_{L^{\infty}(D)}<1$ 
in \textbf{Theorem \ref{basis-dependon-potential}} can be easily satisfied for bounded $v_2$ and small $h$. Therefore, all the results shown 
in Section \ref{sec:NumericalExamples} are based on the one-step enrichment. Moreover, the continuous dependence also shows that
the enrichment will not be necessary if the temporal variation of $v_2$ itself is small, which is also indicated by numerical results.

\subsection{A property of multiscale finite element methods}
\noindent
The following property holds true for the MsFEM and En-MsFEM.
\begin{proposition}[Conservation of total mass]\label{prop:conservation}
Both the MsFEM and En-MsFEM conserve the total mass, i.e.,
	\begin{align}
		||\psi^{k+1}||_{L^2(D)} & = ||\psi^{k}||_{L^2(D)}, \quad \forall k\ge 0. \label{eqn:mass}
	\end{align}
\end{proposition}

\begin{proof}
	By definition, $ \psi_{k}=\sum_{i=1}^{N_H}c_i^k\phi_i$ and $ \psi_{k+1}=\sum_{i=1}^{N_H}c_i^{k+1}\phi_i$. Thus
		\begin{align}
		||\psi^{k+1}||_{L^2(D)}^2 & = \sum_{1\leq i,j\leq N_H} (c_i^{k+1})^* c_j^{k+1} \left(\phi_i,\phi_j\right) = (\bc^{k+1})^* M \bc^{k+1},\label{eqn:density1}\\
		||\psi^{k}||_{L^2(D)}^2 & = \sum_{1\leq i,j\leq N_H} (c_i^{k})^* c_j^{k} \left(\phi_i,\phi_j\right) = (\bc^{k})^* M \bc^{k}.\label{eqn:density2}
		\end{align}
		From the fully discrete scheme \eqref{eqn:fulldiscretescheme0}, we have 
		\begin{align}\label{eqn:ThmMassConservation1}
		i\epsilon M\frac{\bc^{k+1}-\bc^k}{\Delta t}=B\frac{\bc^{k+1}+\bc^k}{2},
		\end{align}
		where $B=\frac{\epsilon^2}{2}S+V_1+V_2(t_{k+\frac{1}{2}})$ is a real symmetric matrix. Multiplying \eqref{eqn:ThmMassConservation1} by $(\frac{\bc^{k+1}+\bc^k}{2})^*$ from the left and using \eqref{eqn:density1} - \eqref{eqn:density2}, we get
		\begin{multline}
		\frac{i\epsilon}{2\Delta t} \left\{||\psi^{k+1}||_{L^2(D)}^2-||\psi^{k}||_{L^2(D)}^2\right\} \\
		+\frac{i\epsilon}{2\Delta t}\left\{(\bc^k)^*M\bc^{k+1}-(\bc^{k+1})^*M\bc^{k}\right\}=(\frac{\bc^{k+1}+\bc^k}{2})^*B(\frac{\bc^{k+1}+\bc^k}{2}).\label{calculation-tmp2}	
	    \end{multline}
		Since $(\frac{\bc^{k+1}+\bc^k}{2})^*B(\frac{\bc^{k+1}+\bc^k}{2})$ on the right-hand side of \eqref{calculation-tmp2} is a real number and
		\begin{align*}
		\Big(\frac{i\epsilon}{2\Delta t}\left\{(\bc^k)^*M\bc^{k+1}-(\bc^{k+1})^*M\bc^{k}\right\}\Big)^*=\frac{i\epsilon}{2\Delta t}\left\{(\bc^k)^*M\bc^{k+1}-(\bc^{k+1})^*M\bc^{k}\right\}
		\end{align*}
		is also real, the imaginary part of \eqref{calculation-tmp2} produces
		\begin{align*}
		\frac{\epsilon}{2\Delta t} \left\{||\psi^{k+1}||_{L^2(D)}^2-||\psi^{k}||_{L^2(D)}^2\right\} = 0,
		\end{align*}
		which concludes the conservation of total mass.
	\end{proof}

It is worth mentioning that in the presence of time-dependent potential, there exists exchange of energy between electron and the external field
and therefore the energy cannot be conserved any more.
	
\section{Numerical examples}\label{sec:NumericalExamples}
\noindent	
In this section, we test the proposed methods for several examples in one and two dimensions. The numerical experiments consist of three 1D examples with a periodic potential, a multiplicative two-scales potential, and a layered two-scales potential, and a 2D example with a checkboard potential. The computational domain is $D=[0,1]$ in 1D and $D=[0,1]\times[0,1]$ in 2D and the final time is $T=1$ in all examples. In all cases, if not specified,
we denote $ \psi^{\epsilon}_{\textrm{ref}} $ the reference solution obtained by the Crank-Nicolson scheme in time with a very small stepsize $\tau=\frac{1}{2^{20}}\approx  9.5\times10^{-7}$ and the standard FEM in space with a very small meshsize $h=\frac{1}{3\times2^{15}}\approx 1.0\times10^{-5}$. We also show the performance of standard FEM for comparison. We denote $ \psi^{\epsilon}_{\textrm{num}} $ the numerical solutions obtained by any coarse mesh methods (standard FEM, MsFEM or En-MsFEM). In all examples, the total mass is checked to be a constant during the time evolution.

The initial data in 1D and 2D are chosen as
\begin{equation*}
\psi_{\textrm{in}}(x)=(\frac{1}{2\pi\sigma^2})^{1/4}e^{-\frac{(x-0.5)^2}{4\sigma^2}},~\sigma = 0.2,
\end{equation*}
and
\begin{equation*}
\psi_{\textrm{in}}(x,y)=(\frac{1}{2\pi\sigma^2})^{1/2}e^{\frac{-(x-1/2)^2-(y-1/2)^2}{4\sigma^2}},~\sigma = 0.2,
\end{equation*}
respectively.

In what follows, we shall compare the relative error between the numerical solutions and the reference solution in both $ L^2 $ norm and $ H^1 $ norm
\begin{align*}
\textrm{Error}_{L^2} =\dfrac{||\psi^{\epsilon}_{\textrm{num}}-\psi^{\epsilon}_{\textrm{ref}}||_{L^2}}
{||\psi^{\epsilon}_{\textrm{ref}}||_{L^2}}, \quad 
\textrm{Error}_{H^1} =\dfrac{||\psi^{\epsilon}_{\textrm{num}}-\psi^{\epsilon}_{\textrm{ref}}||_{H^1}}
{||\psi^{\epsilon}_{\textrm{ref}}||_{H^1}}, 
\end{align*}
where the $ L^2 $ norm and $ H^1 $ norm are defined as
\[
||\psi^{\epsilon}||^2_{L^2}=\int_{D}|\psi^{\epsilon}|^2 \d\bx, \quad ||\psi^{\epsilon}||^2_{H^1}=\int_{D}|\nabla \psi^{\epsilon}|^2 \d\bx + \int_{D}|\psi^{\epsilon}|^2 \d\bx,
\]
respectively. Note that $||\psi^{\epsilon}_{\textrm{ref}}||_{H^1}$ increases significantly as $\veps$ reduces. We therefore consider relative errors in both $L^2$ norm and $H^1$ norm.

Moreover, we will also check the performance of our methods for the computation of observables, including the position density
\begin{align}\label{eqn:density}
n^{\epsilon}(\bx,t)=|\psi^{\epsilon}(\bx,t)|^2,
\end{align}
and the energy density
\begin{align}\label{eqn:energy}
e^{\epsilon}(\bx,t)=\frac{\epsilon^2}{2}|\nabla \psi^{\epsilon}(\bx,t)|^2+
\big(v_1^{\epsilon}(\bx)+v_2(\bx, t)\big)|\psi^{\epsilon}(\bx,t)|^2.
\end{align}

\begin{example}[1D case with a spatially periodic potential and a sine type time-dependent potential]\label{example1}
In this experiment, the potential $v^\epsilon(x,t) = v_1^\epsilon(x) + v_2(x,t)$. We start with the so-called Mathieu model, where $ v_1^{\epsilon}(x) =\cos(2\pi\frac{x}{\epsilon})$ is a periodic function of $x/\epsilon$. The time-dependent part of the potential is $v_2(x,t) = E_0\sin(2\pi t)x$ with $E_0=20$. We set $\epsilon=\frac{1}{32}$.
	
In Figures \ref{fig:ex1L2err_wave_final} and \ref{fig:ex1H1err_wave_final}, we record the relative $L^2$ and $H^1$ errors on a series of coarse meshes when $ H = \frac{1}{64}, \frac{1}{96},\frac{1}{128}, \frac{1}{192}, \frac{1}{256}, \frac{1}{384} $. 
Multiscale basis functions in the En-MsFEM combines the basis functions used in MsFEM and the enriched basis obtained when the maximum of $v_2(\bx,t)$ is achieved. The number of enriched basis is $1/8$ of that in MsFEM, and thus the computational complexity of En-MsFEM is approximately the same as that of the MsFEM. We choose $\Delta t = 4\tau=\frac{1}{2^{18}}$ so the approximation error due to the temporal discretization can be ignored.

In Figures \ref{fig:ex1L2err_density_final} and \ref{fig:ex1L2err_engdensity_final}, we show the relative $L^2$ errors of the position density and energy density functions. From these results, for moderate coarse meshes, we can see that MsFEM reduces the approximation error by more than two orders of magnitude than that of the standard FEM in both $L^2$ and $H^1$ norms. In addition, En-MsFEM further reduces the error by another one order of magnitude in $L^2$ norm and by several times in $H^1$ norm. Figure \ref{fig:ex1_err_series} further illustrates how the approximation error is reduced as time evolves.
\begin{figure}[htbp]
	\centering
	\begin{subfigure}{0.45\textwidth}
		\includegraphics[width=1.0\textwidth]{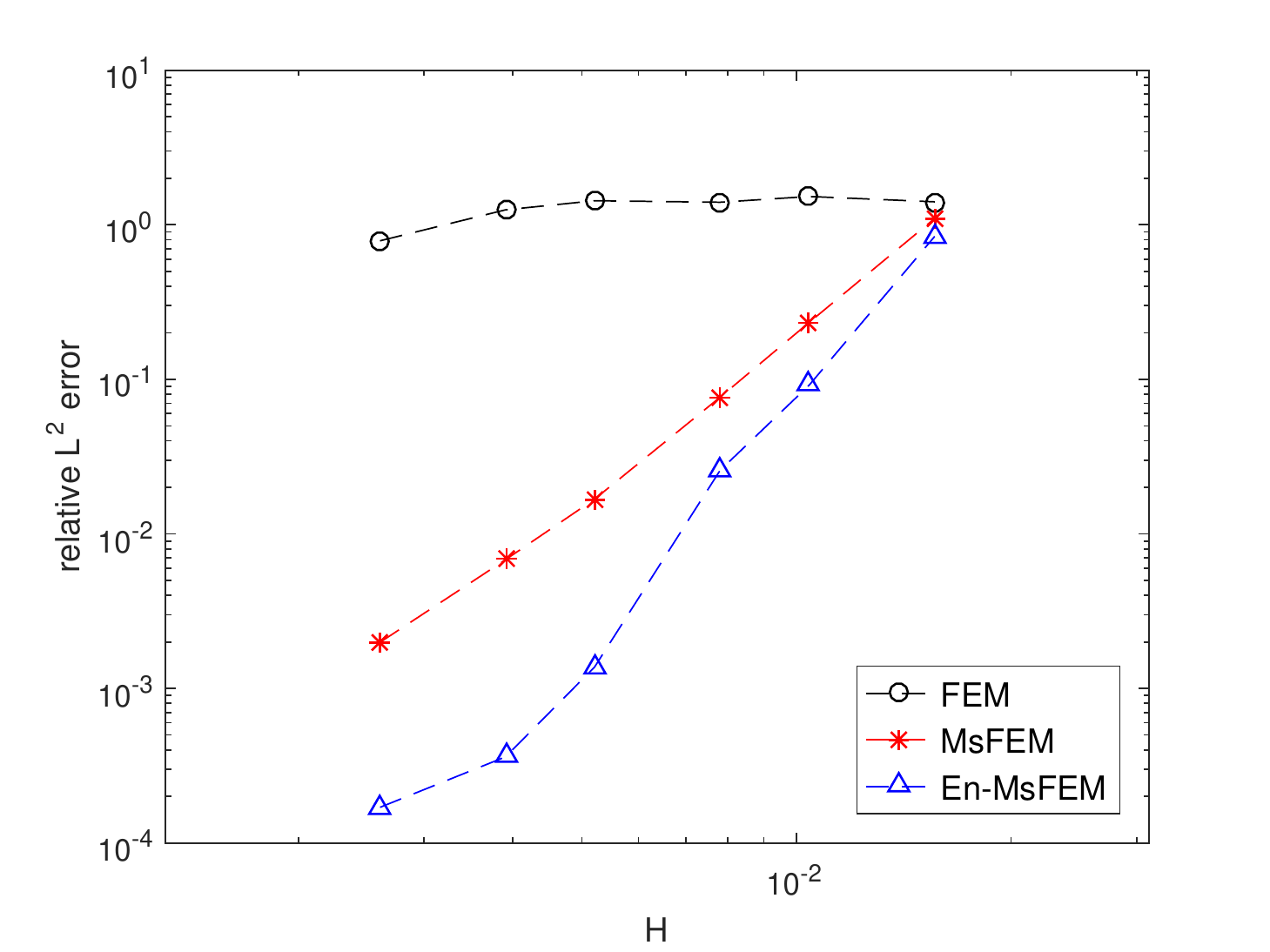}
		\caption{$L^2$ error of the wavefunction}
		\label{fig:ex1L2err_wave_final}
	\end{subfigure}
	\begin{subfigure}{0.45\textwidth}
		\centering
		\includegraphics[width=\textwidth]{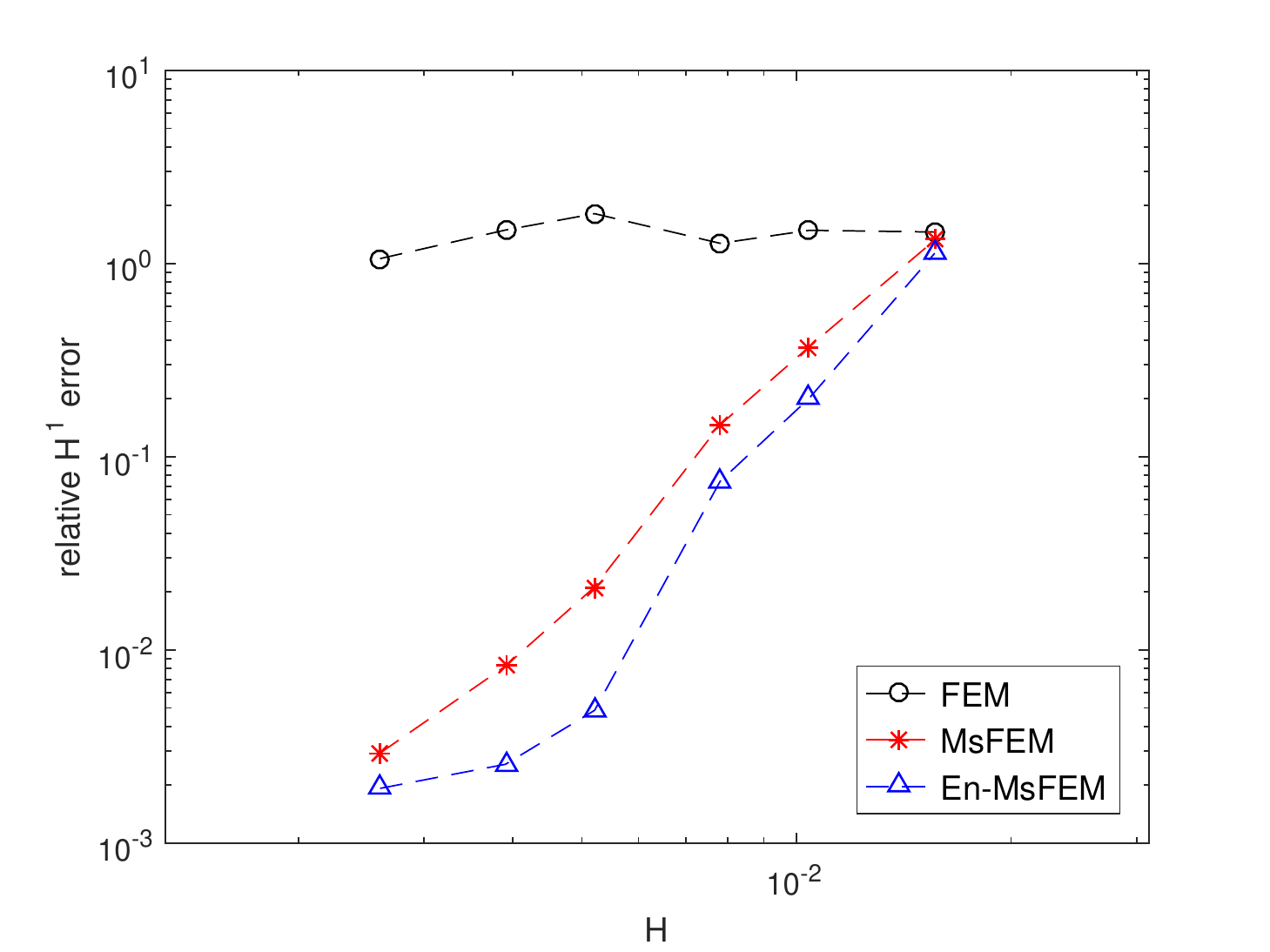}
		\caption{$H^1$ error of the wavefunction}
		\label{fig:ex1H1err_wave_final}
	\end{subfigure}
	\caption{Relative errors of the wavefunction at $T=1$ in \textbf{Example \ref{example1}}.}
	\label{fig:ex1_wavefuncerr_final_v3}
\end{figure}

\begin{figure}[htbp]
	\centering
	\begin{subfigure}{0.45\textwidth}
		\includegraphics[width=1.0\textwidth]{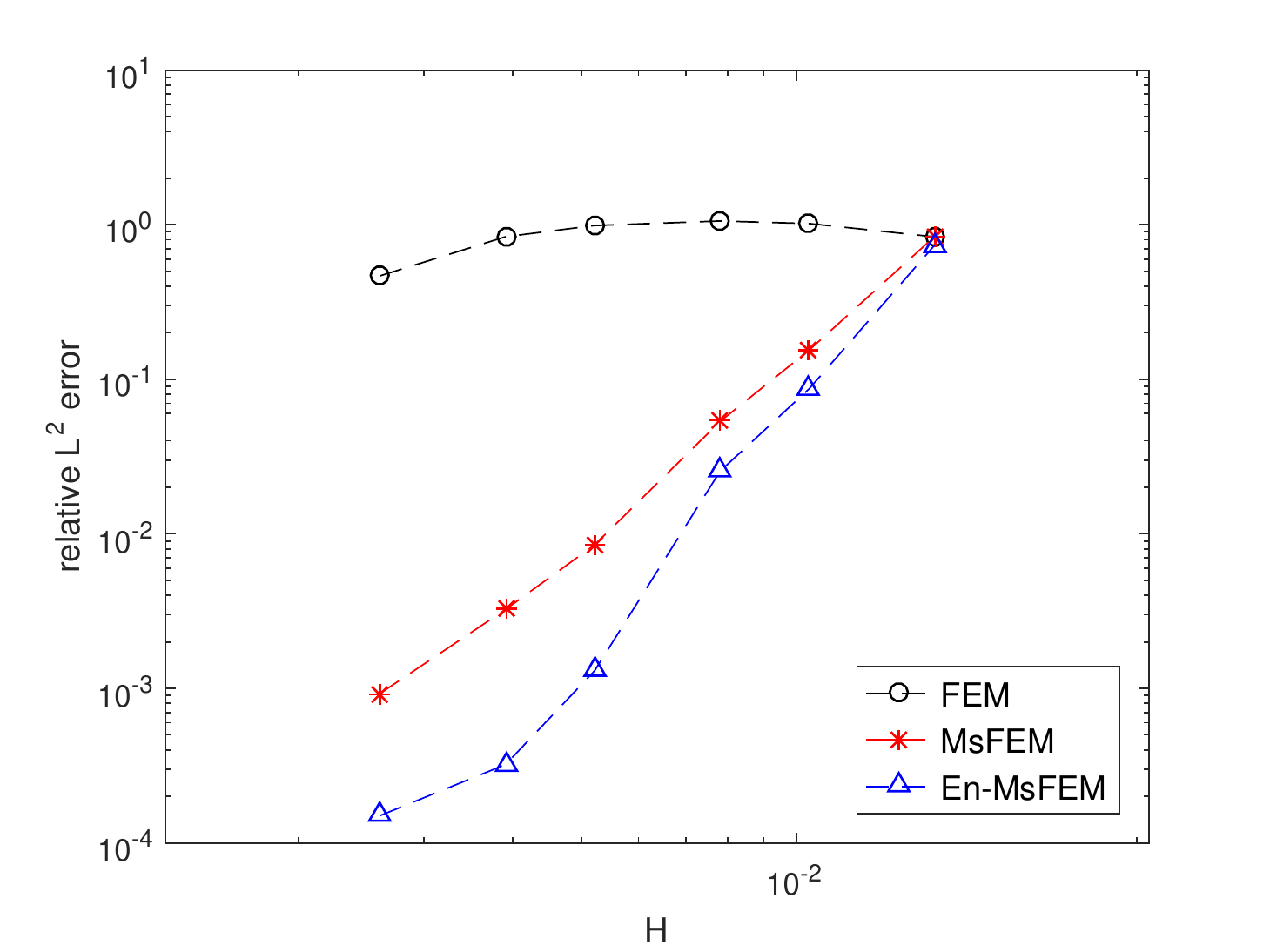}
		\caption{$L^2$ error of the position density function}
		\label{fig:ex1L2err_density_final}
	\end{subfigure}
	\begin{subfigure}{0.45\textwidth}
		\includegraphics[width=1.0\textwidth]{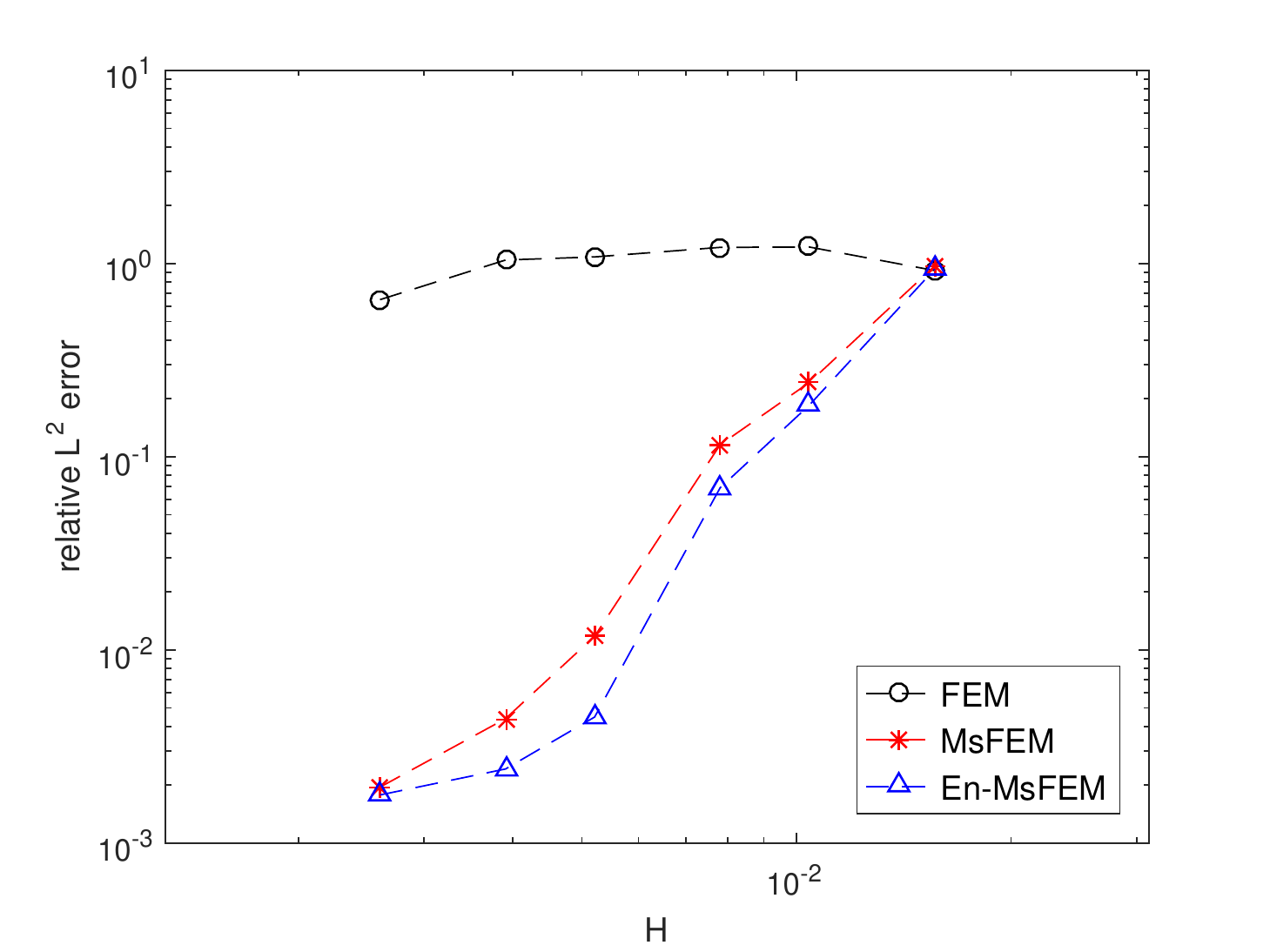}
		\caption{$L^2$ error of the energy density function}
		\label{fig:ex1L2err_engdensity_final}
	\end{subfigure}
	\caption{Relative errors of density functions at $T=1$ in \textbf{Example \ref{example1}}.}
	\label{fig:ex1_densityerr_final_v3}
\end{figure}

\begin{figure}[htbp]
	\centering
	\begin{subfigure}{0.45\textwidth}
		\includegraphics[width=1.0\textwidth]{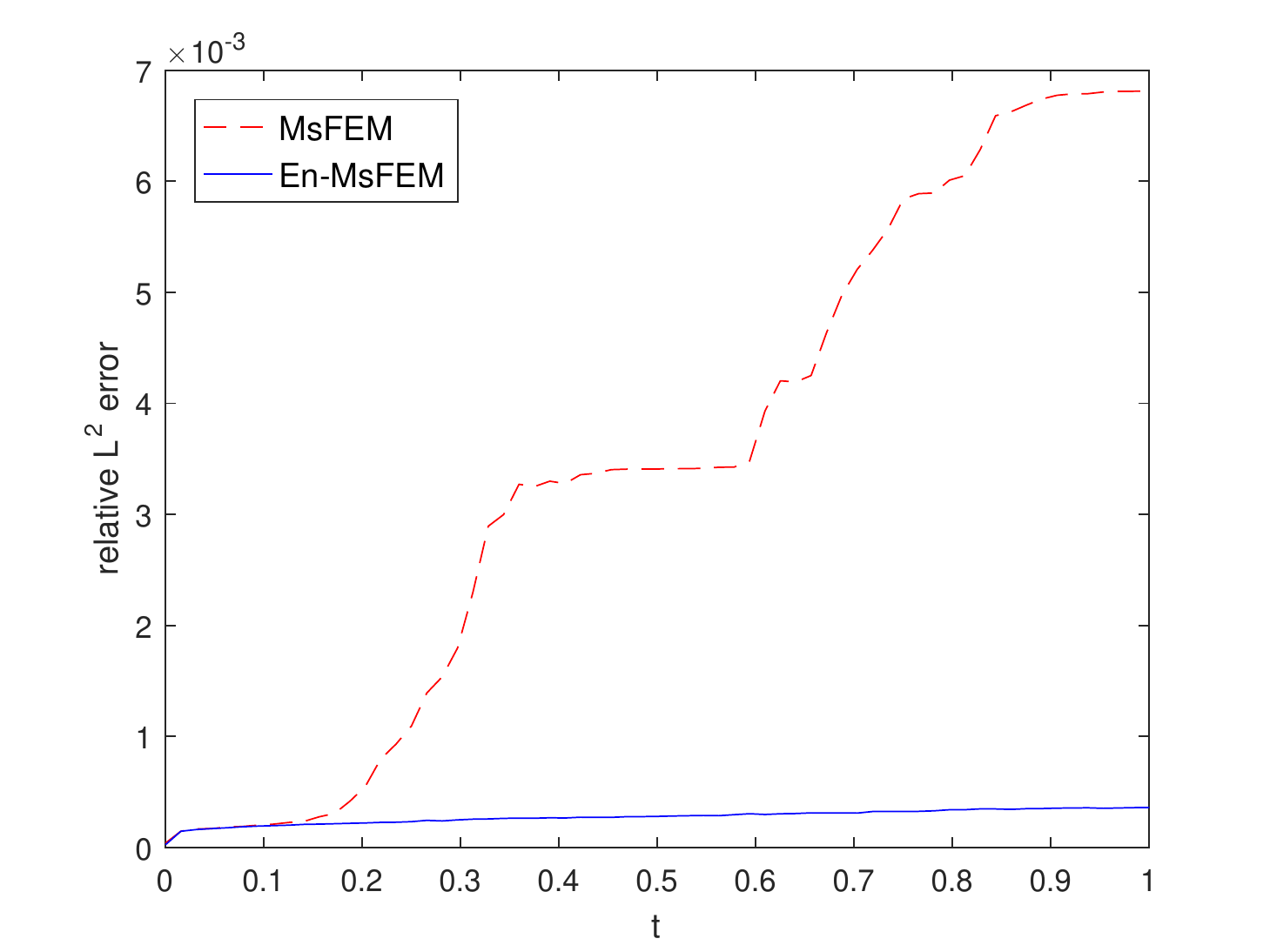}
		\caption{$L^2$ error of the wavefunction}
	\end{subfigure}
	\begin{subfigure}{0.45\textwidth}
		\centering
		\includegraphics[width=1.0\textwidth]{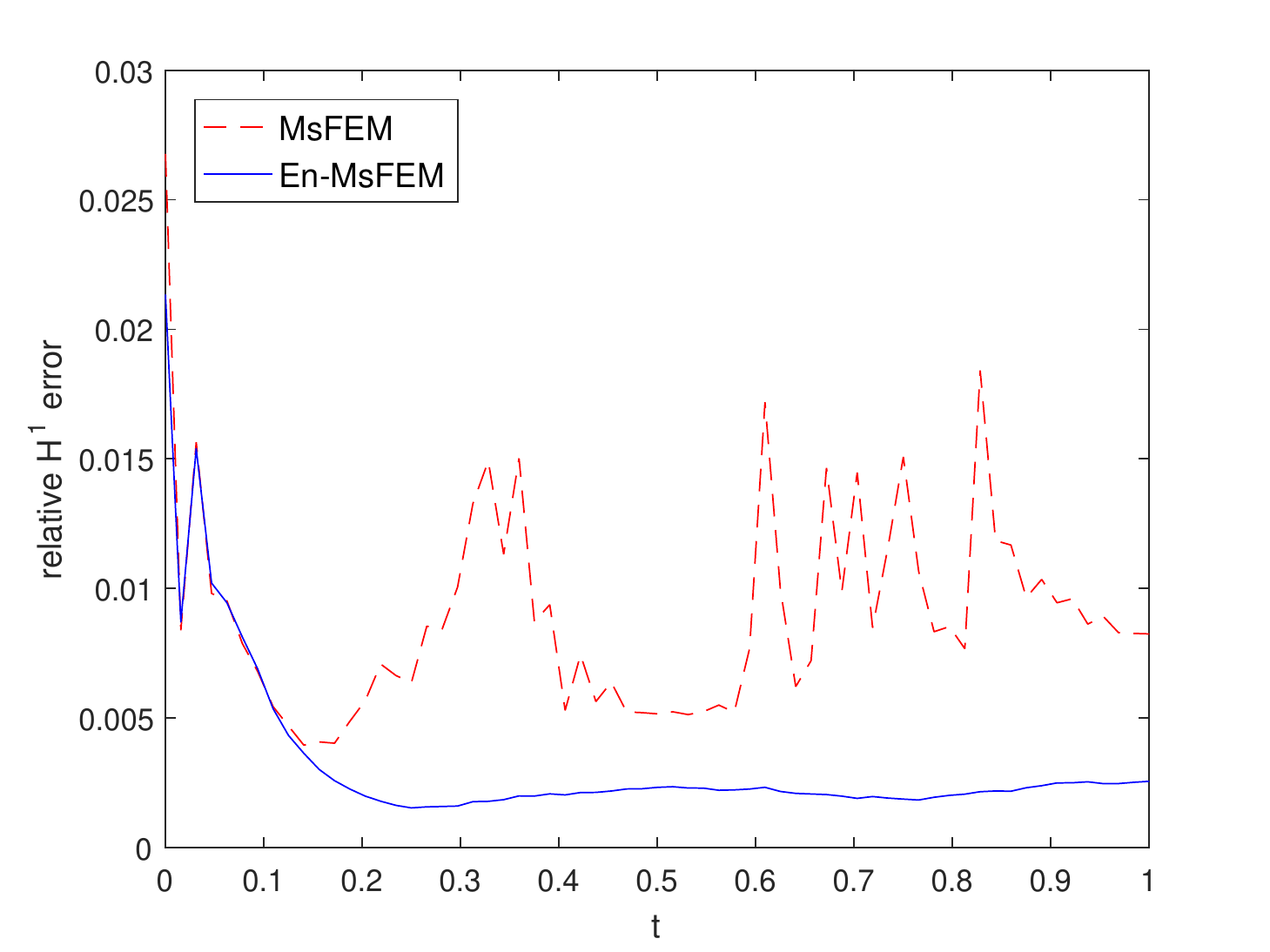}
		\caption{$H^1$ error of the wavefunction}
	\end{subfigure}
	\begin{subfigure}{0.45\textwidth}
		\includegraphics[width=1.0\textwidth]{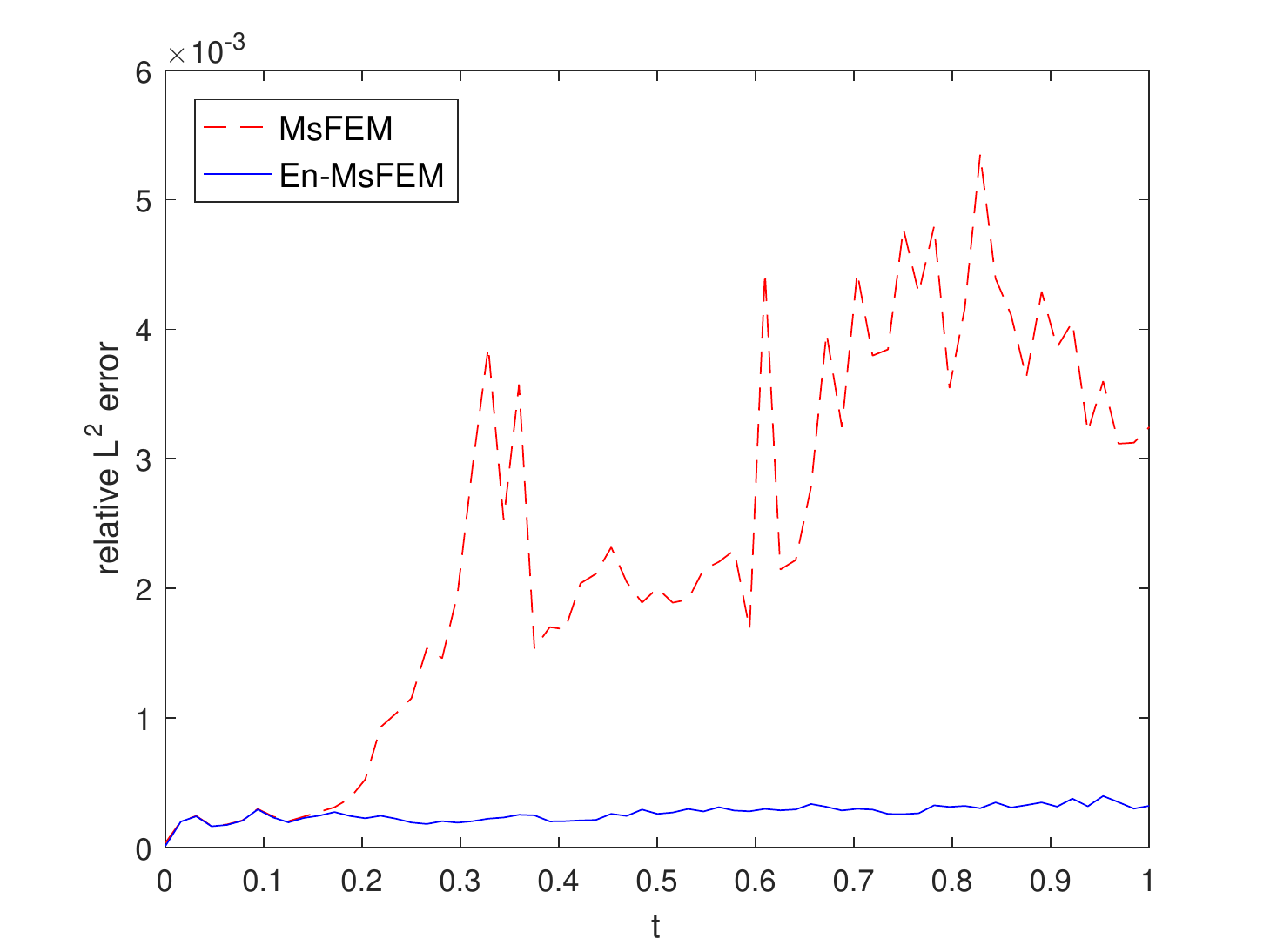}
		\caption{$L^2$ error of the position density function}
	\end{subfigure}
	\begin{subfigure}{0.45\textwidth}
		\includegraphics[width=1.0\textwidth]{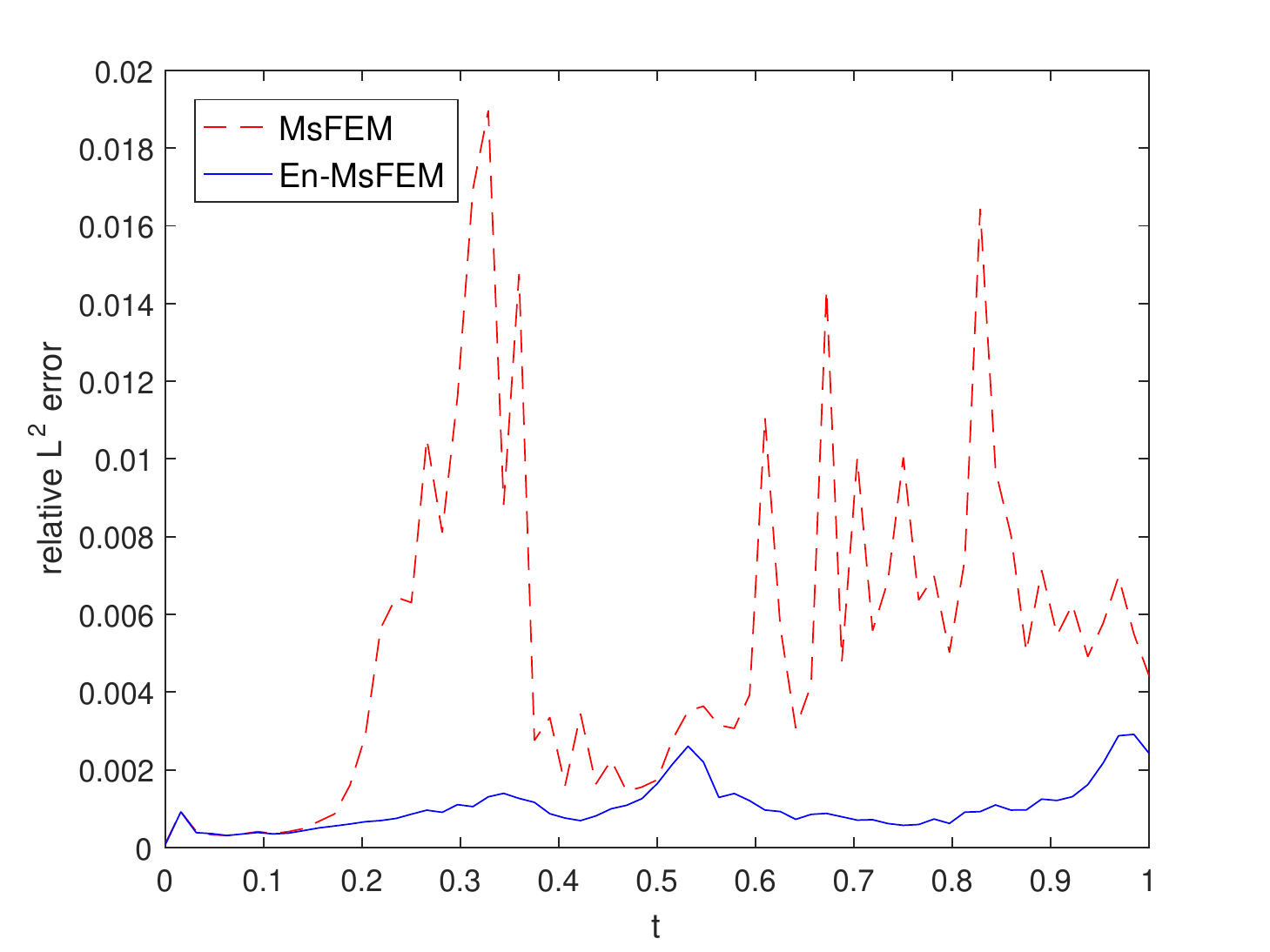}
		\caption{$L^2$ error of the energy density function}
	\end{subfigure}
	\caption{Relative errors of MsFEM and En-MsFEM as a function of time when $H=\frac{1}{256}$ in \textbf{Example \ref{example1}}.}
	\label{fig:ex1_err_series}
\end{figure}

We find that the En-MsFEM is superior in the case when the magnitude of the time-dependent potential is large, i.e., large $E_0$ 
in $v_2(x,t)$. Moreover, its efficiency is not affected by the magnitude of $\epsilon$. Even if $\epsilon$ is further reduced, the En-MsFEM still
performs well as long as $E_0$ is large.
\end{example}
 
\begin{example}[1D case with a multiplicative two-scale potential and a periodic time-dependent potential]\label{example2}
In this experiment, the potential $v^\epsilon(x,t) = v_1^\epsilon(x) + v_2(x,t)$. The time-independent part of the potential $v_1^{\epsilon}(x)=\sin(2x^2)\sin(2\pi \frac{x}{\epsilon})$ is a multiplicative two-scale potential. The time-dependent part is $v_2(x,t) = E_0\frac{\exp(2\sin(2\pi t))-1}{\exp(2)-1}x$ with $E_0=20$.

Set $\epsilon=\frac{1}{32}$. We compute numerical solutions on a series of coarse meshes $H=\frac{1}{48}$, $\frac{1}{64}$, $\frac{1}{96}$, $\frac{1}{128}$, $\frac{1}{192}$, $\frac{1}{256}$ in the MsFEM and the number of enriched basis is $\frac{1}{8}$ of that in the MsFEM, obtained at the time when $v_2(\bx,t)$ is maximized. We choose $\Delta t = 4\tau=\frac{1}{2^{18}}$ so the approximation error due to the temporal discretization can be ignored.
	
In Figure \ref{fig:ex2_waveerr_final_v3} we plot relative $L^2$ and $H^1$ errors of the standard FEM, MsFEM, and En-MsFEM at the final time $T=1$. In Figure \ref{fig:ex2_densityerr_final_v3}, we show the relative $L^2$ errors of
the position density and energy density functions. From these results, for moderate coarse meshes, we can see that MsFEM reduces the approximation error by more than two orders of magnitude than that of the standard FEM in both $L^2$ and $H^1$ norms. In addition, En-MsFEM further reduces the error by another one order of magnitude in $L^2$ norm and by several times in $H^1$ norm.

\begin{figure}[htbp]
	\centering
	\begin{subfigure}{0.45\textwidth}
		\includegraphics[width=1.0\textwidth]{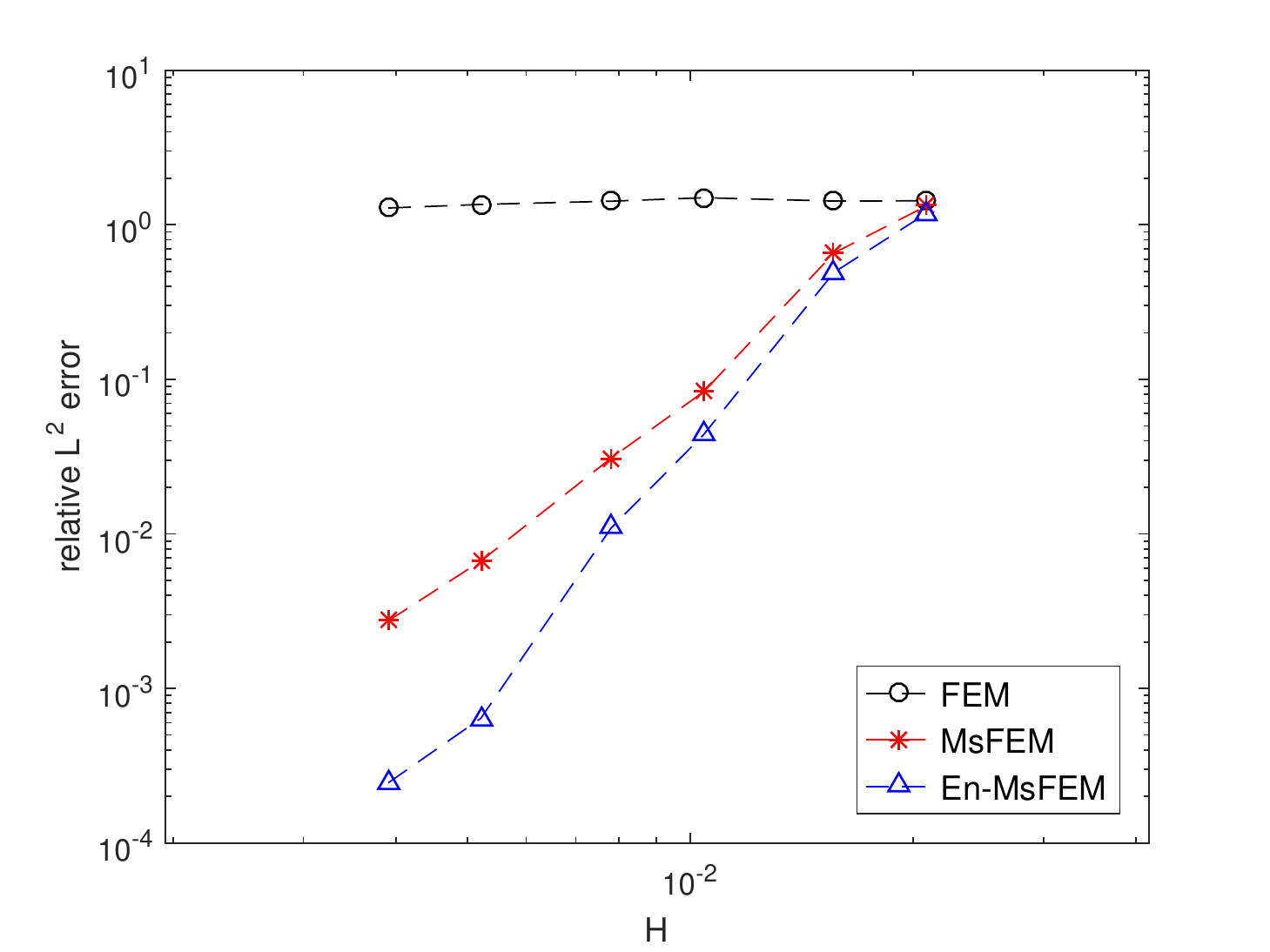}
		\caption{$L^2$ error of the wavefunction}
		\label{fig:ex2L2err_wave_final}
	\end{subfigure}
	\begin{subfigure}{0.45\textwidth}
		\centering
		\includegraphics[width=\textwidth]{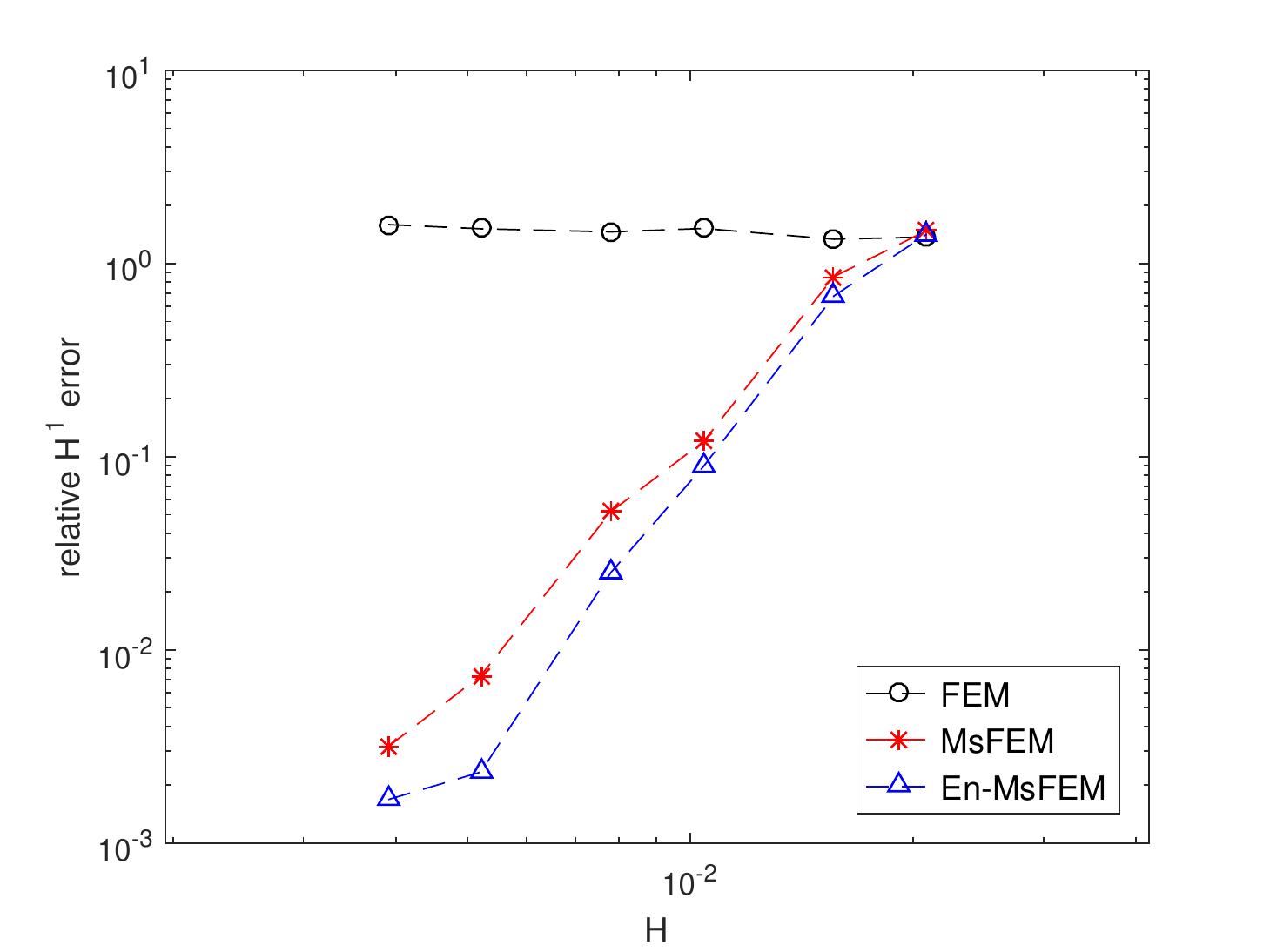}
		\caption{$H^1$ error of the wavefunction}
		\label{fig:ex2H1err_wave_final}
	\end{subfigure}
	\caption{Relative errors of wavefunction at $T=1$ in \textbf{Example \ref{example2}}. }
	\label{fig:ex2_waveerr_final_v3}
\end{figure}

\begin{figure}[htbp]
	\centering
	\begin{subfigure}{0.45\textwidth}
		\includegraphics[width=1.0\textwidth]{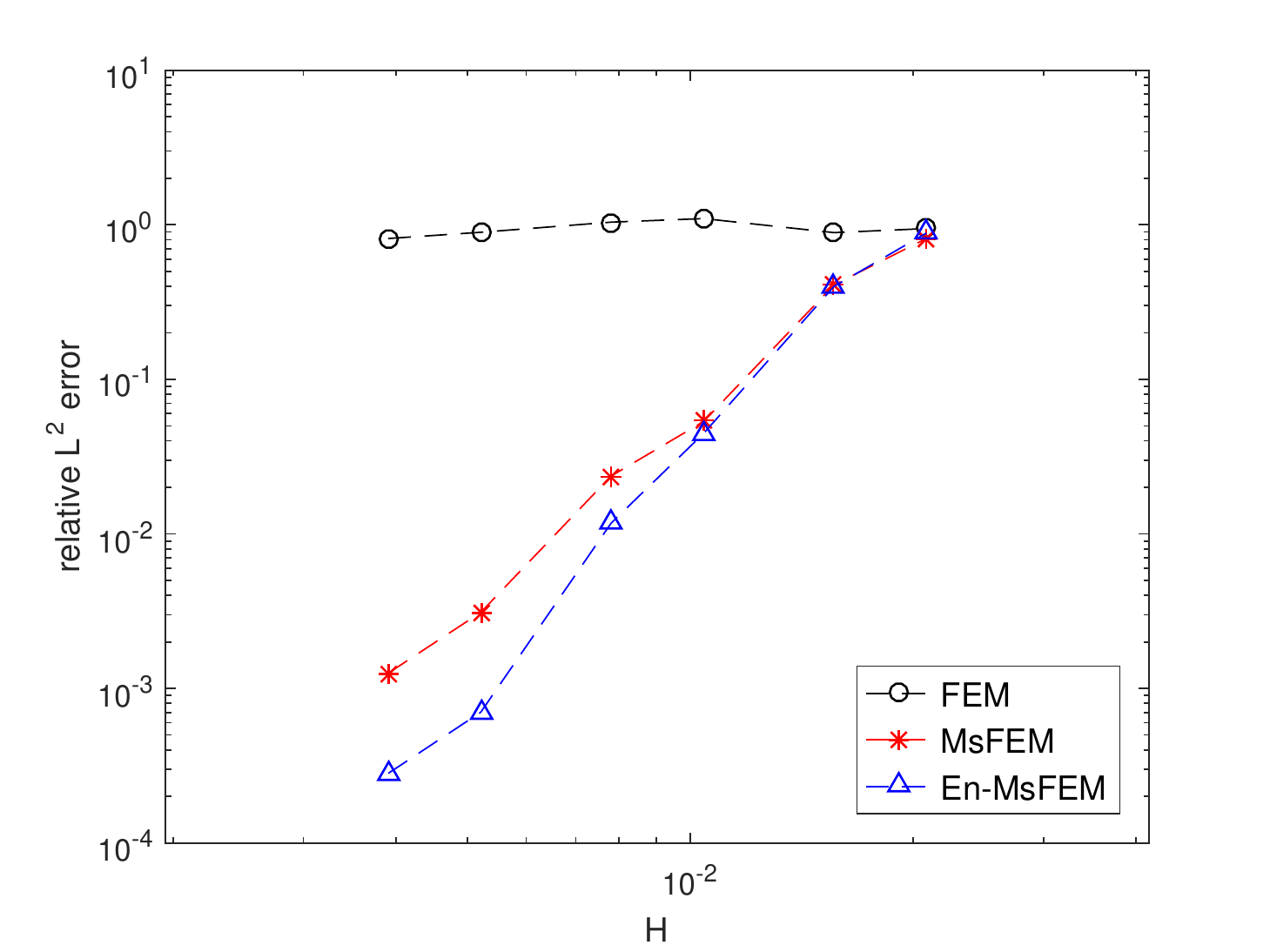}
		\caption{$L^2$ error of the position  density function}
	\end{subfigure}
	\begin{subfigure}{0.45\textwidth}
		\includegraphics[width=1.0\textwidth]{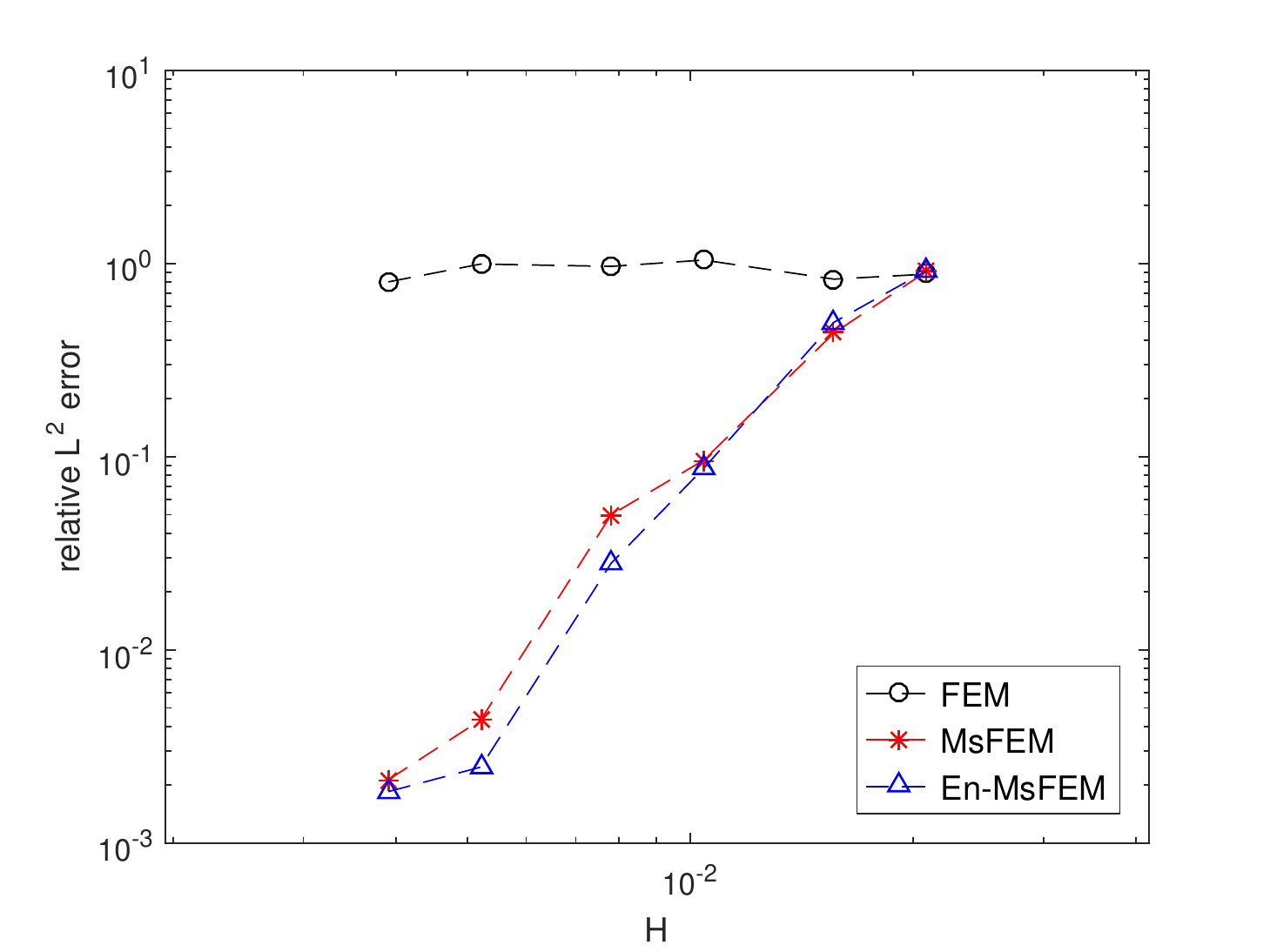}
		\caption{$L^2$ error of the energy density function}
	\end{subfigure}
	\caption{Relative errors of density functions at $T=1$ in \textbf{Example \ref{example2}}.}
	\label{fig:ex2_densityerr_final_v3}
\end{figure}

We visualize profiles of the position density and energy density functions of the standard FEM, MsFEM, and En-MsFEM in Figure \ref{fig:ex2_density_engdensity}.
Nice agreement is observed. We visualize the time evolution of total mass, total energy, and energy difference of MsFEM and En-MsFEM in Figure \ref{fig:ex2_totaleng_series}. The total mass is conserved, which agrees with \textbf{Proposition \ref{prop:conservation}}.
Due to the energy exchange in the presence of an external field, the total energy is not conserved. The energy difference is small in the MsFEM and En-MsFEM further reduces the difference by two orders of magnitude as time evolves.

\begin{figure}[htbp]
	\centering
	\begin{subfigure}{0.45\textwidth}
		\includegraphics[width=1.0\textwidth]{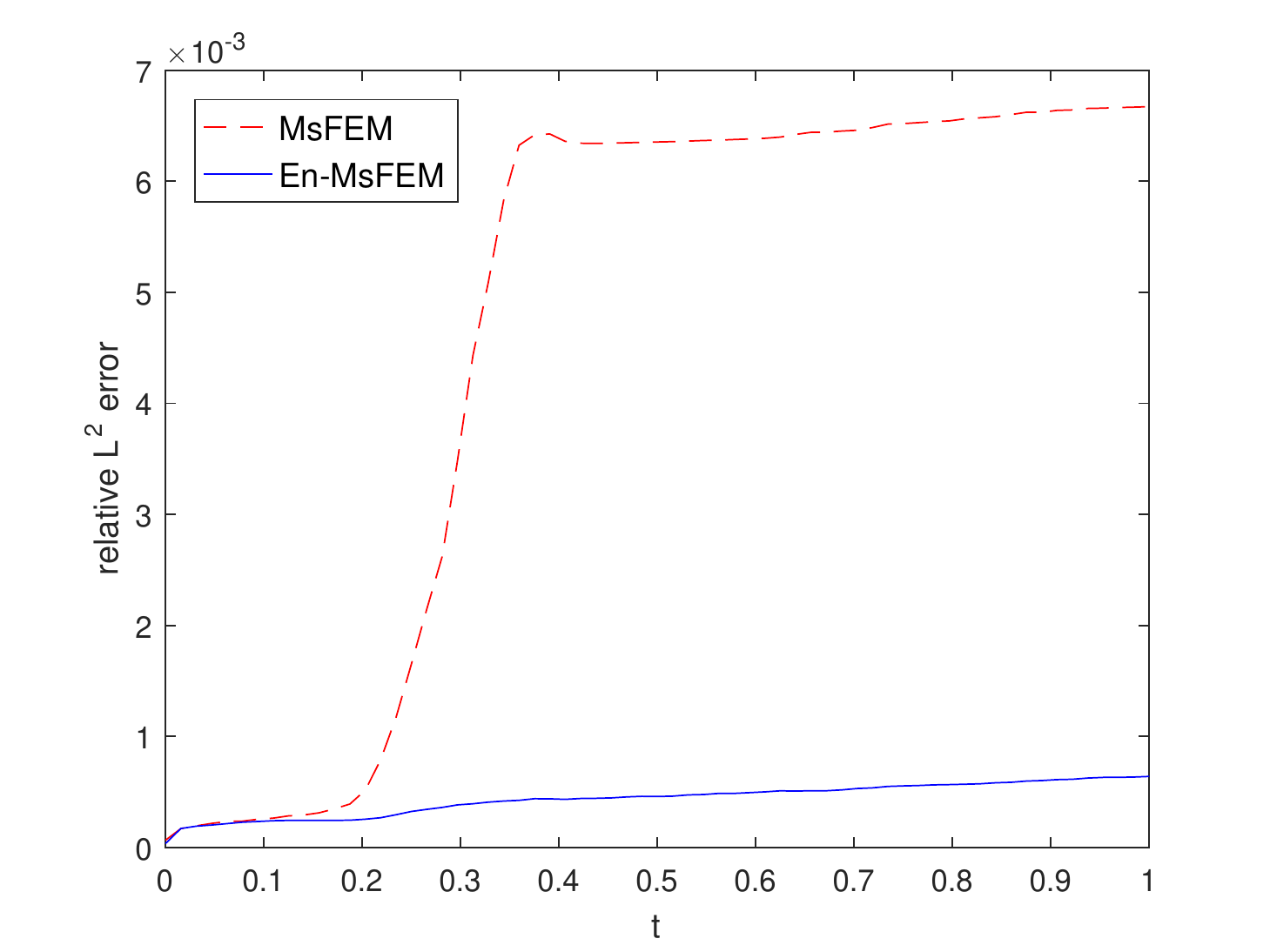}
		\caption{$L^2$ error of the wavefunction}
	\end{subfigure}
	\begin{subfigure}{0.45\textwidth}
		\centering
		\includegraphics[width=\textwidth]{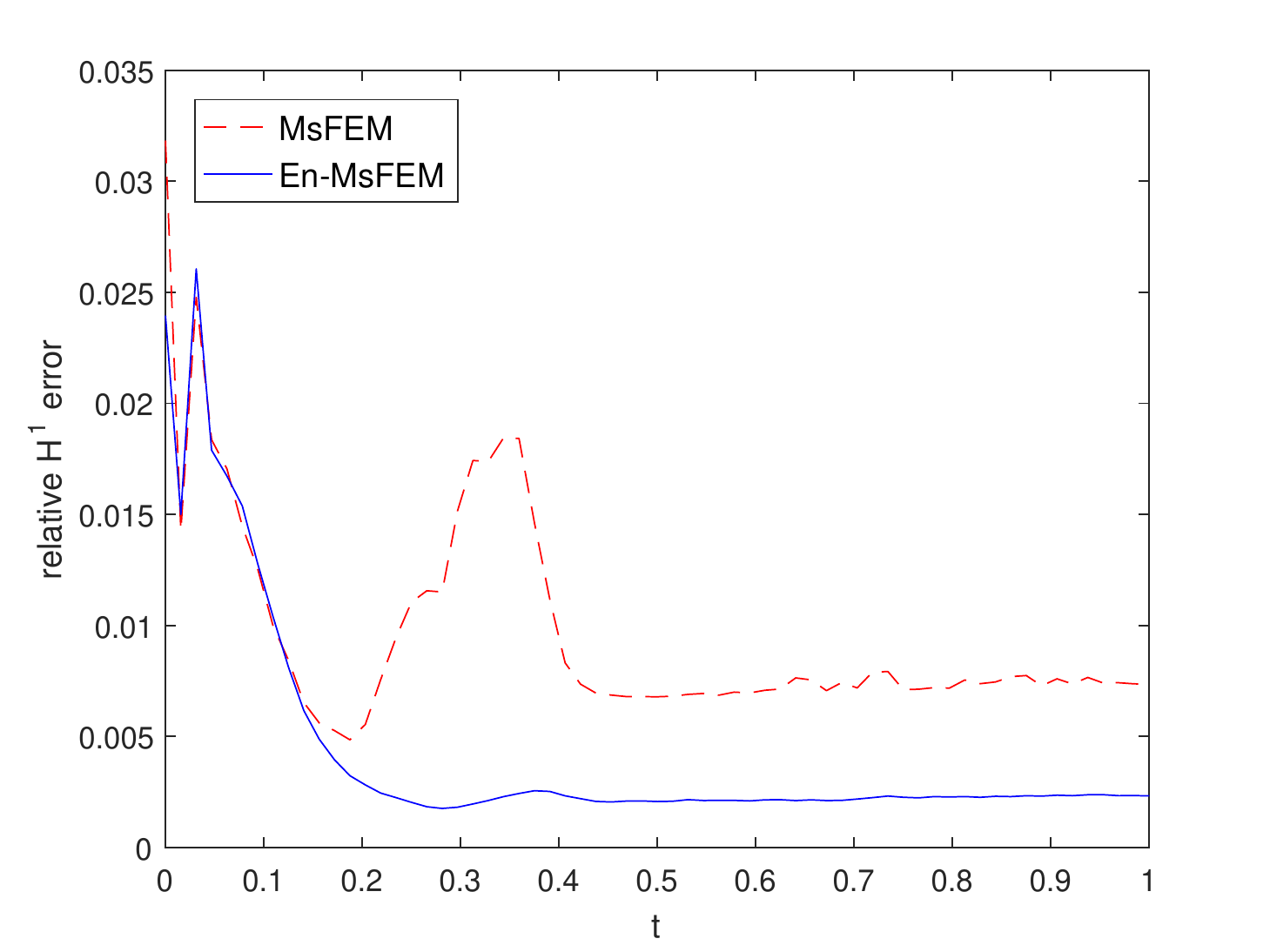}
		\caption{$H^1$ error of the wavefunction}
	\end{subfigure}
	\begin{subfigure}{0.45\textwidth}
		\includegraphics[width=1.0\textwidth]{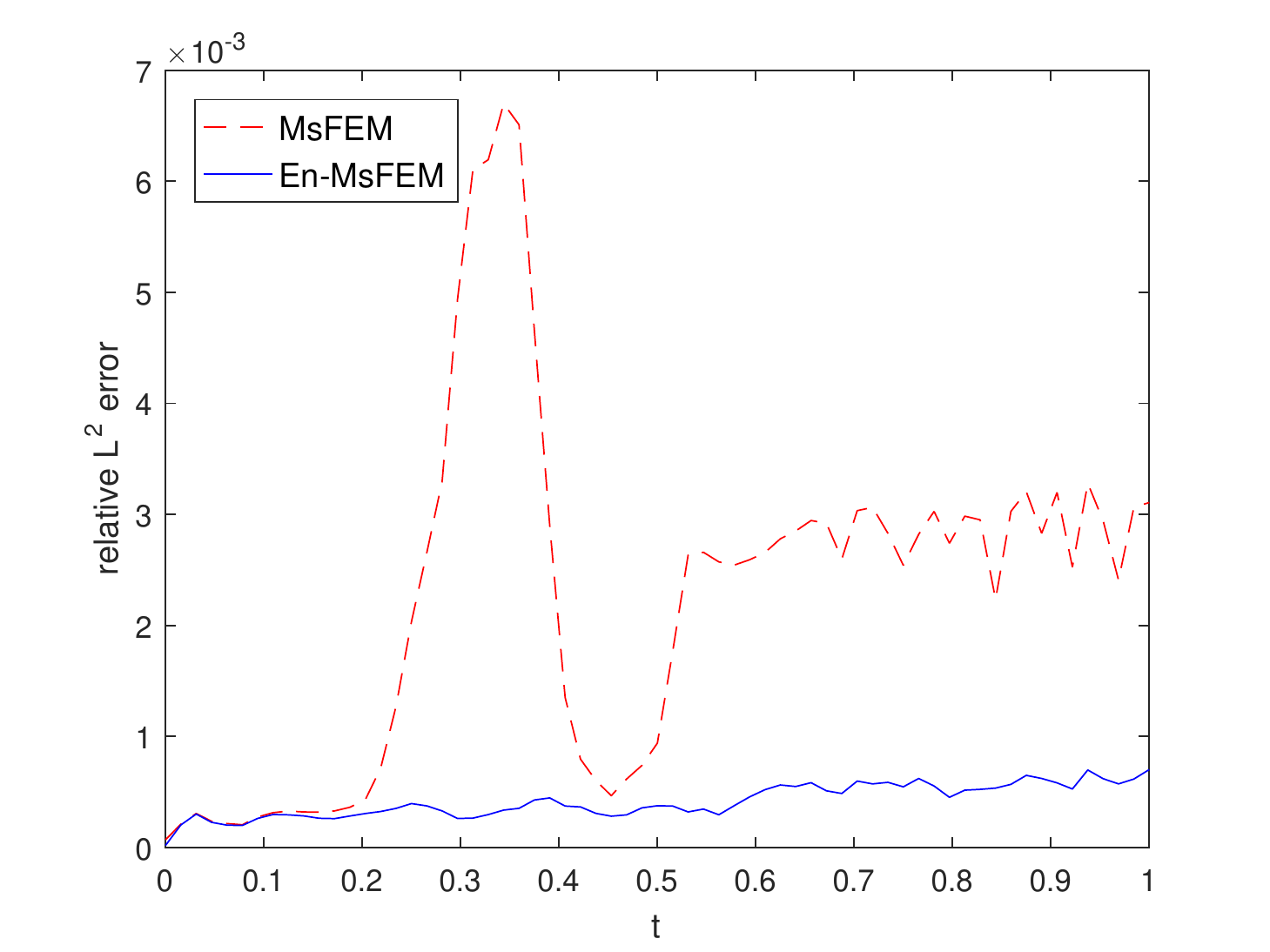}
		\caption{$L^2$ error of the position density function}
	\end{subfigure}
	\begin{subfigure}{0.45\textwidth}
		\includegraphics[width=1.0\textwidth]{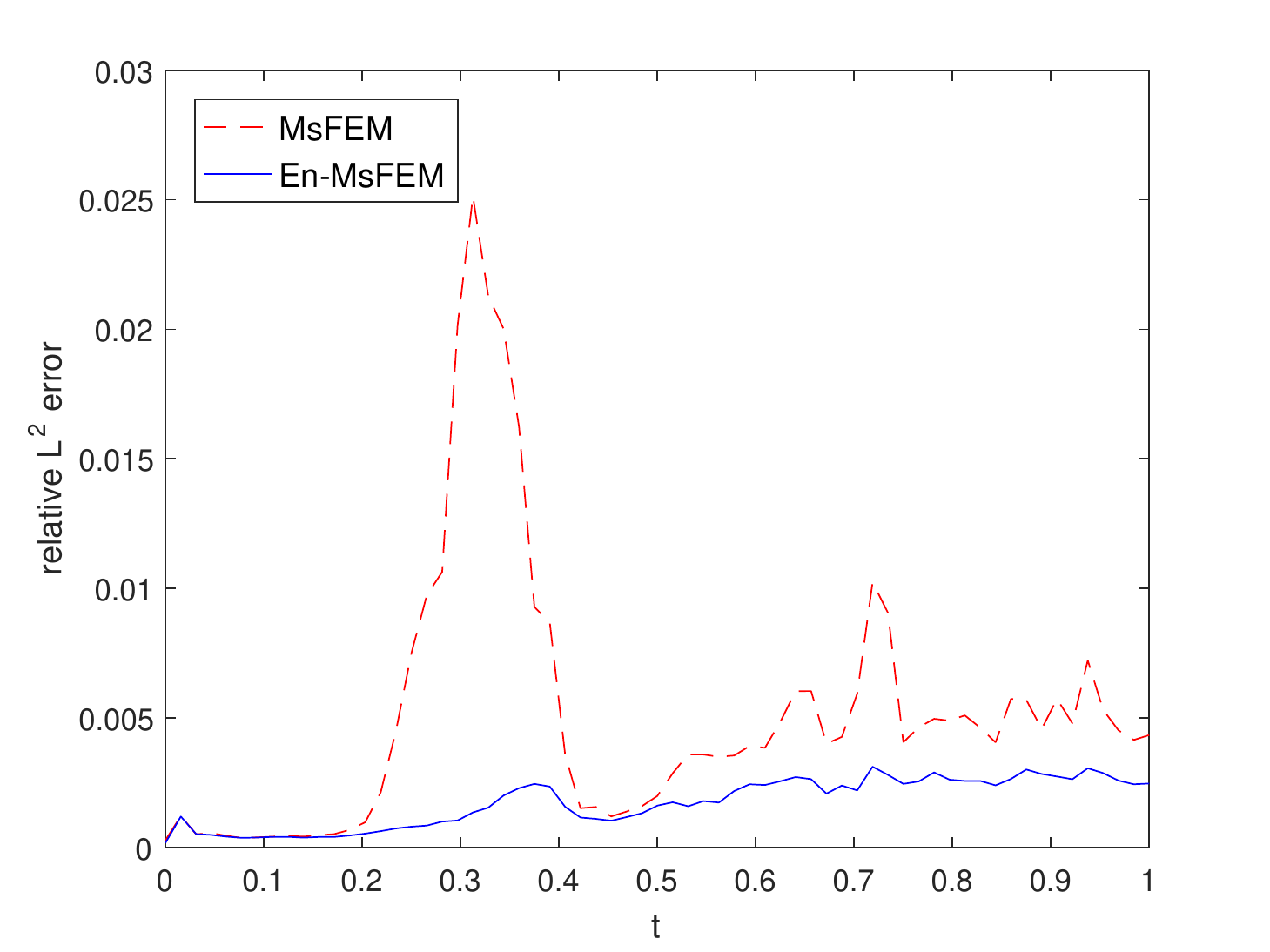}
		\caption{$L^2$ error of the energy density function}
	\end{subfigure}
	\caption{ Relative errors of MsFEM and En-MsFEM as a function of time when $H=\frac{1}{192}$ in \textbf{Example \ref{example2}}.}
	\label{fig:ex2_err_series}
\end{figure}

\begin{figure}[htbp]
	\centering
	\begin{subfigure}{0.45\textwidth}
		\includegraphics[width=1.0\textwidth]{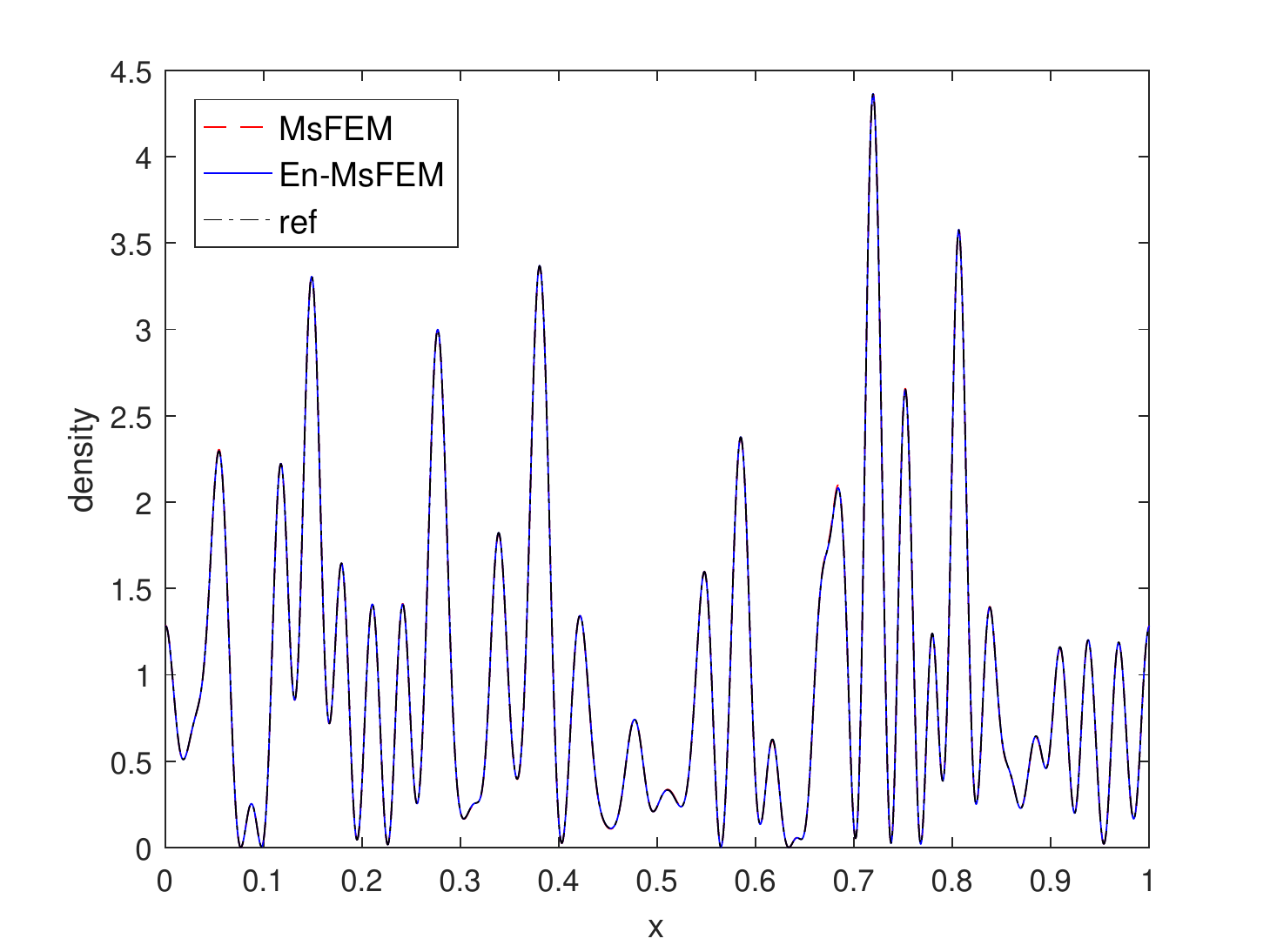}
		\caption{$n_{\textrm{num}}^{\epsilon}(\bx,T)$ and $n_{\textrm{ref}}^{\epsilon}(\bx,T)$}
	\end{subfigure}
	\begin{subfigure}{0.45\textwidth}
		\centering
		\includegraphics[width=\textwidth]{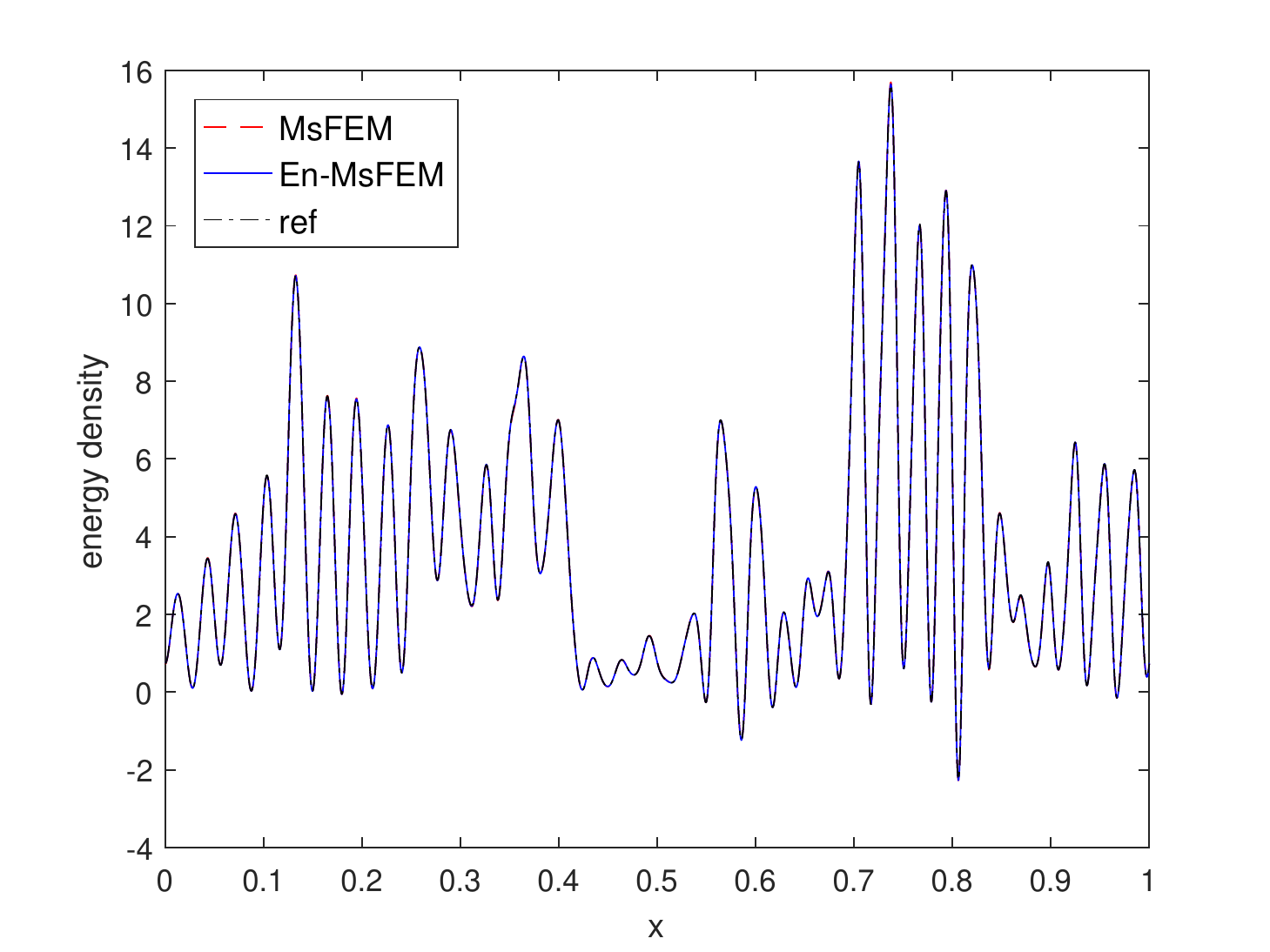}
		\caption{$e_{\textrm{num}}^{\epsilon}(\bx,T)$ and $e_{\textrm{ref}}^{\epsilon}(\bx,T)$}
	\end{subfigure}
	\begin{subfigure}{0.45\textwidth}
		\includegraphics[width=1.0\textwidth]{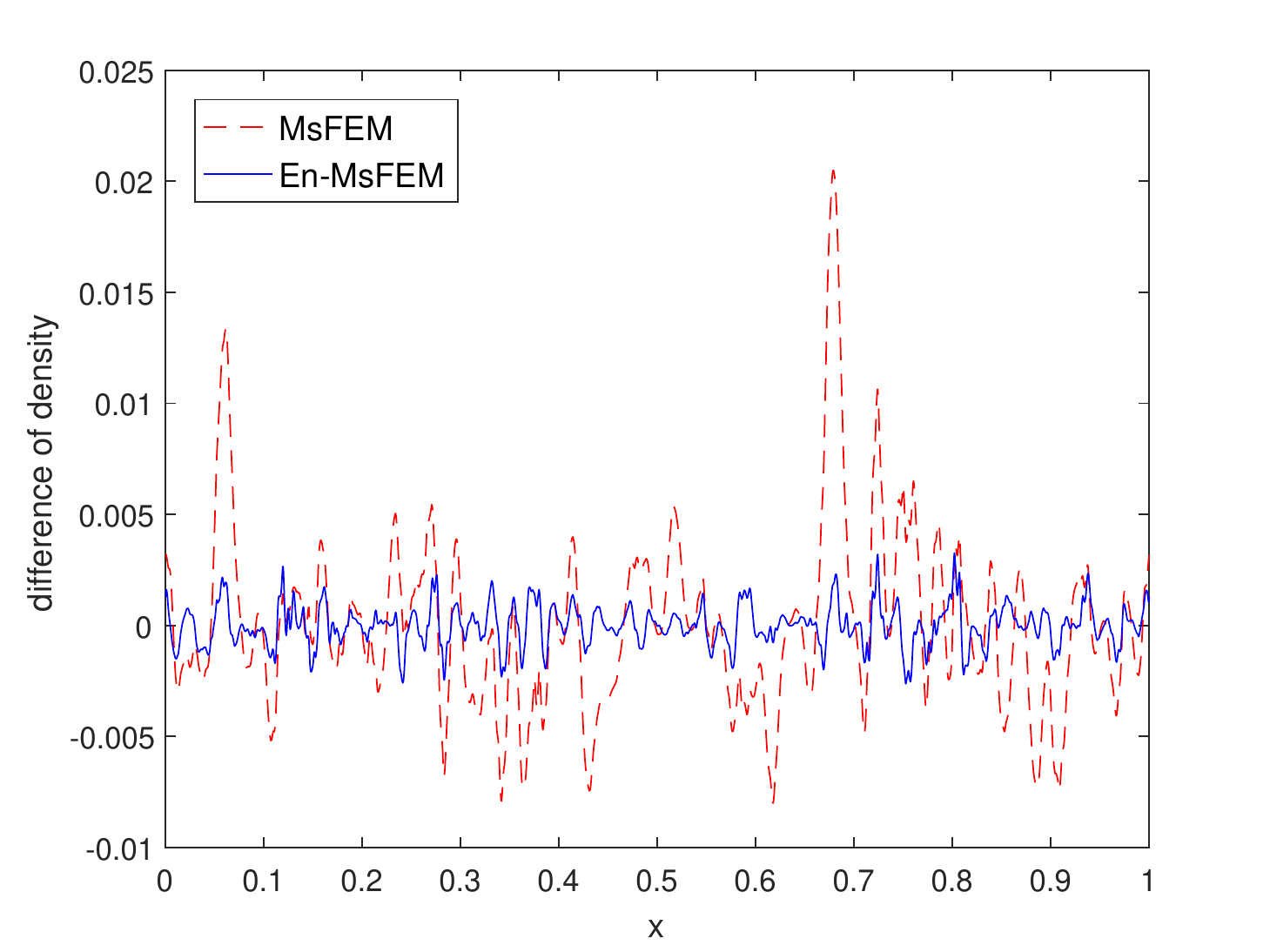}
		\caption{$n_{\textrm{num}}^{\epsilon}(\bx,T)-n_{\textrm{ref}}^{\epsilon}(\bx,T)$}
	\end{subfigure}
	\begin{subfigure}{0.45\textwidth}
		\centering
		\includegraphics[width=\textwidth]{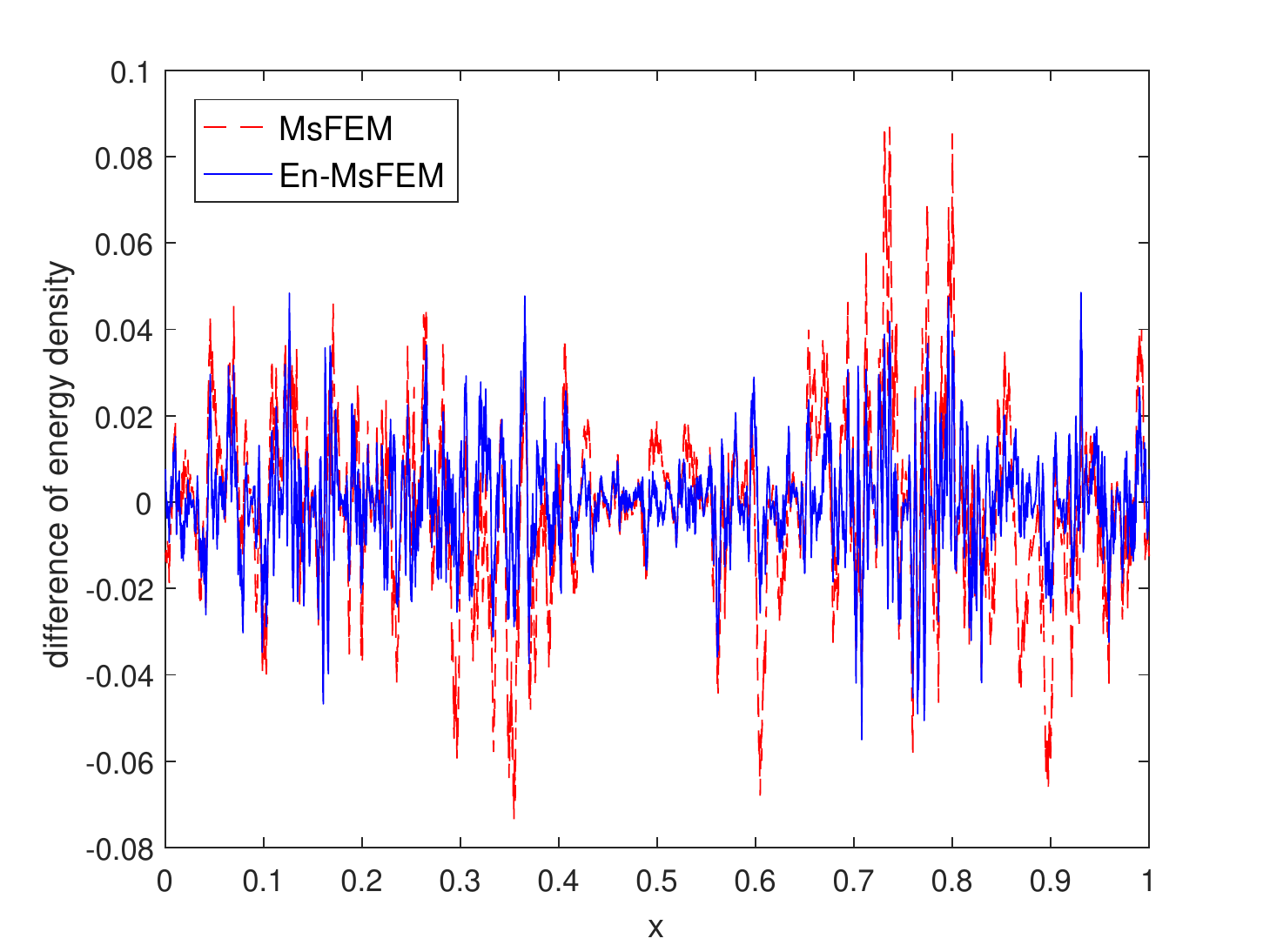}
		\caption{$e_{\textrm{num}}^{\epsilon}(\bx,T)-e_{\textrm{ref}}^{\epsilon}(\bx,T)$}
	\end{subfigure}
	\caption{Profiles of the position density and energy density functions and differences at $T=1$ in \textbf{Example \ref{example2}} when $H=\frac{1}{192}$.}
	\label{fig:ex2_density_engdensity}
\end{figure} 

\begin{figure}[htbp]
	\centering
	\begin{subfigure}{0.32\textwidth}
		\includegraphics[width=1.0\textwidth]{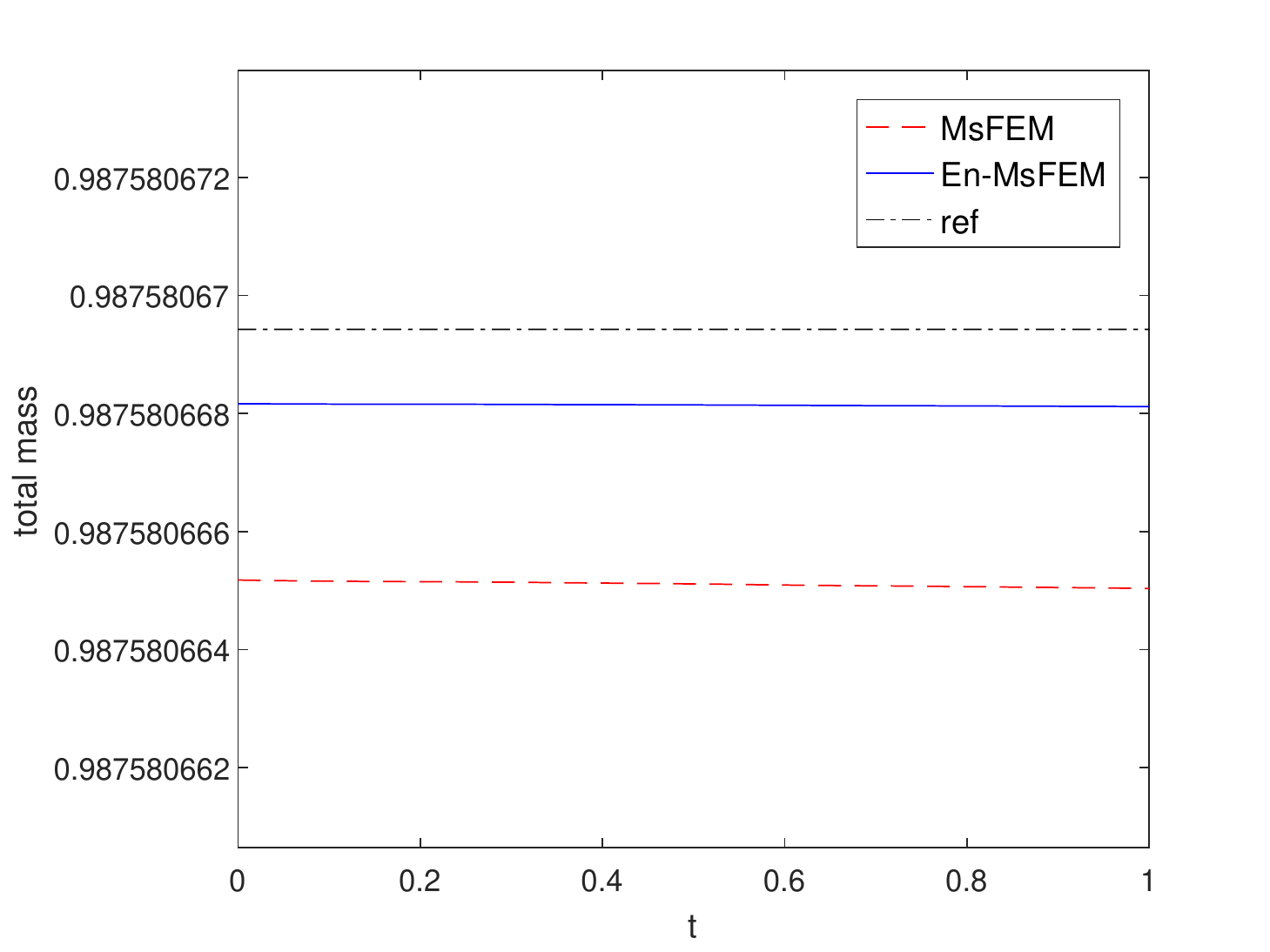}
		\caption{Total mass}
	\end{subfigure}
	\begin{subfigure}{0.32\textwidth}
		\includegraphics[width=1.0\textwidth]{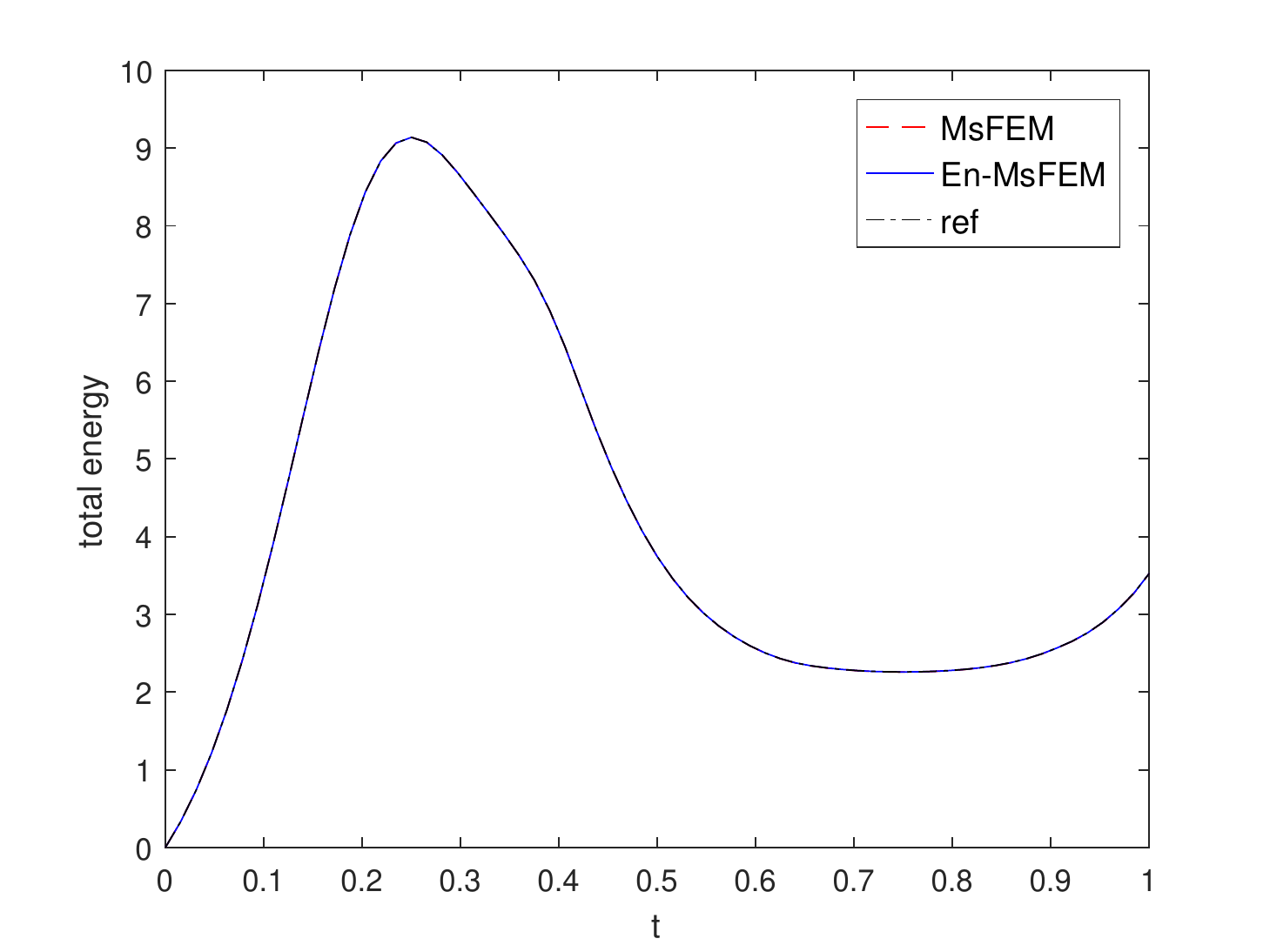}
		\caption{Total energy}
	\end{subfigure}
	\begin{subfigure}{0.32\textwidth}
		\centering
		\includegraphics[width=\textwidth]{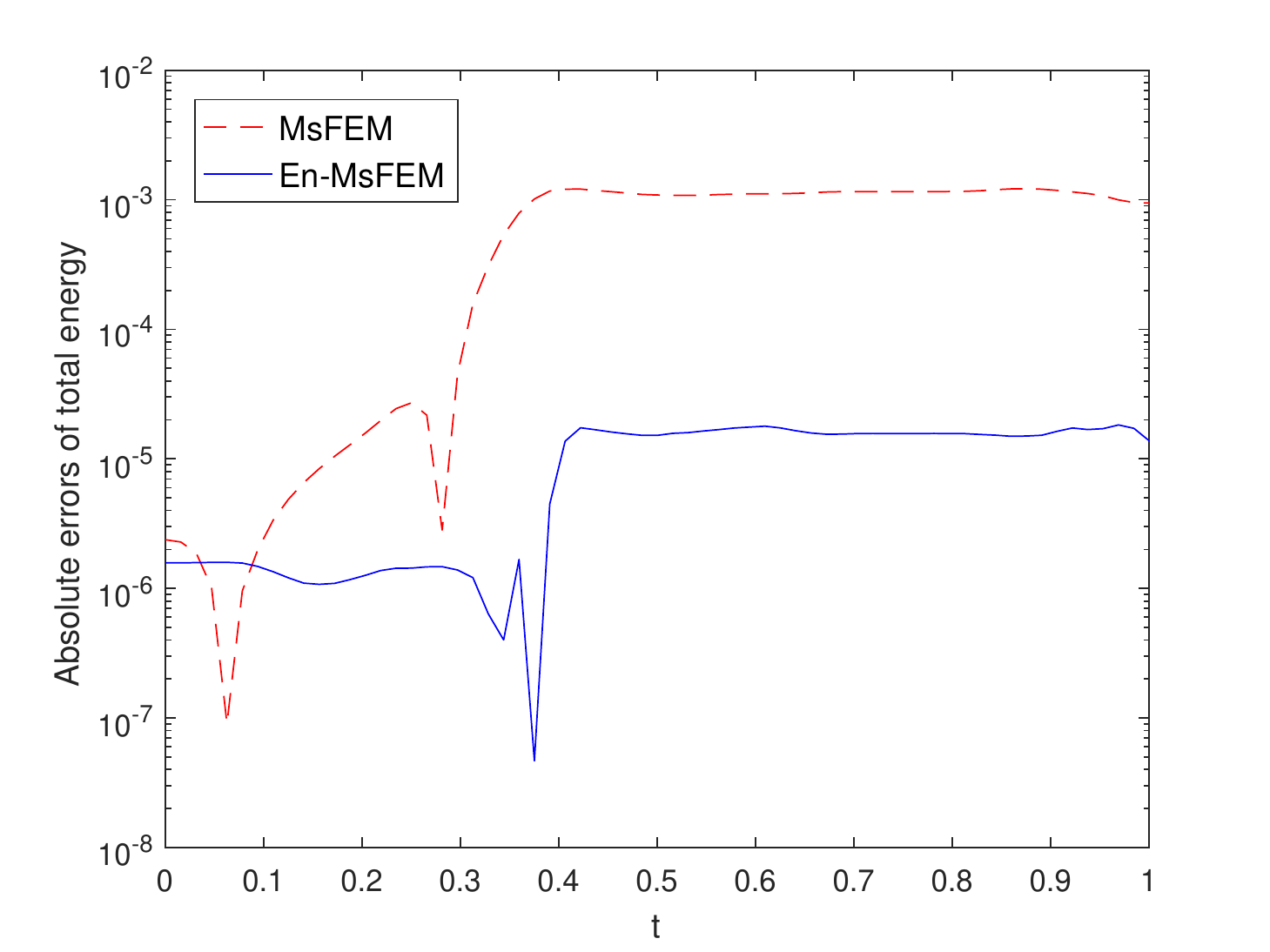}
		\caption{Energy difference}
	\end{subfigure}
	\caption{Time evolution of total mass, total energy and energy differences of energy using the MsFEM and En-MsFEM in \textbf{Example \ref{example2}} when $H=\frac{1}{192}$.}
	\label{fig:ex2_totaleng_series}
\end{figure}
\end{example}

\begin{example}[1D case with a layered potential]\label{example3}
In this experiment, the potential $v^\epsilon(x,t) = v_1^\epsilon(x) + v_2(x,t)$. We set 
\begin{equation*}\label{eqn:potential3}
v_1^{\epsilon}(x)= 2(x-0.5)^2 - \frac{1}{2} + \left\{
\begin{aligned}
\frac{1}{2}\cos(2\pi\frac{x}{\epsilon}),\qquad&0\leq x\leq \frac{1}{2},\\
\frac{1}{2}\cos(2\pi\frac{x}{\epsilon_2})+\frac{1}{2},\qquad&\frac{1}{2}< x\leq 1,
\end{aligned}
\right.
\end{equation*}
where $ \epsilon=1/32 $ and $ \epsilon_2=1/24 $. The time-dependent part over one period is 
\begin{equation*}
v_2(x,t) = E_0x\times\left\{
\begin{aligned}
4t,\qquad&0\leq t\leq \frac{1}{4},\\
2-4t,\qquad&\frac{1}{4}< x\leq \frac{1}{2},
\end{aligned}
\right.
\end{equation*}
where $E_0=20$. 

We set $\epsilon=\frac{1}{32}$ and compute numerical solutions on a series of coarse meshes $H=\frac{1}{64}$, $\frac{1}{96}$, $\frac{1}{128}$, $\frac{1}{192}$, $\frac{1}{256}$, $\frac{1}{384}$ in MsFEM and
the number of enriched basis is $\frac{1}{8}$ of that in the MsFEM, obtained at the time when $v_2(\bx,t)$ is maximized. We choose $\Delta t = 4\tau=\frac{1}{2^{18}}$ so the approximation error due to the temporal discretization can be ignored.

In Figure \ref{fig:ex3_waveerr_final_v3} we plot relative $L^2$ and $H^1$ errors of the standard FEM, MsFEM, and En-MsFEM at the final time $T=1$. In Figure \ref{fig:ex3_densityerr_final_v3} we plot relative $L^2$ errors of density functions by using standard FEM, MsFEM, and En-MsFEM at the final time $T=1$.  
In Figure \ref{fig:ex3_err_series}, we plot relative $L^2$ errors of wavefunction, positive density function, and energy density function as time evolves. From these numerical results, we find the the performance of the MsFEM 
and En-MsFEM is the same as previous two examples.

\begin{figure}[htbp]
	\centering
	\begin{subfigure}{0.45\textwidth}
		\includegraphics[width=1.0\textwidth]{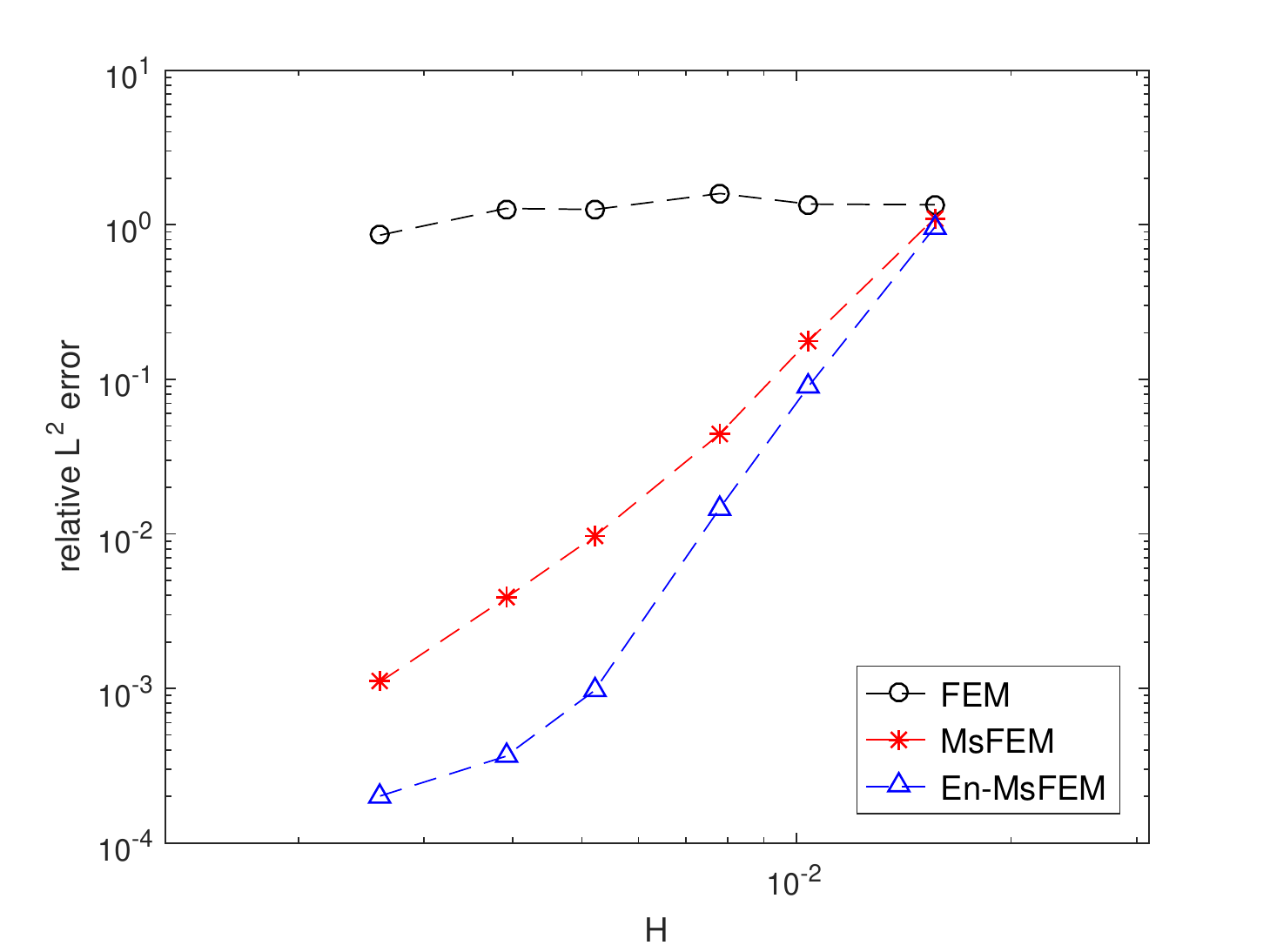}
		\caption{$L^2$ error of the wavefunction}
	\end{subfigure}
	\begin{subfigure}{0.45\textwidth}
		\centering
		\includegraphics[width=\textwidth]{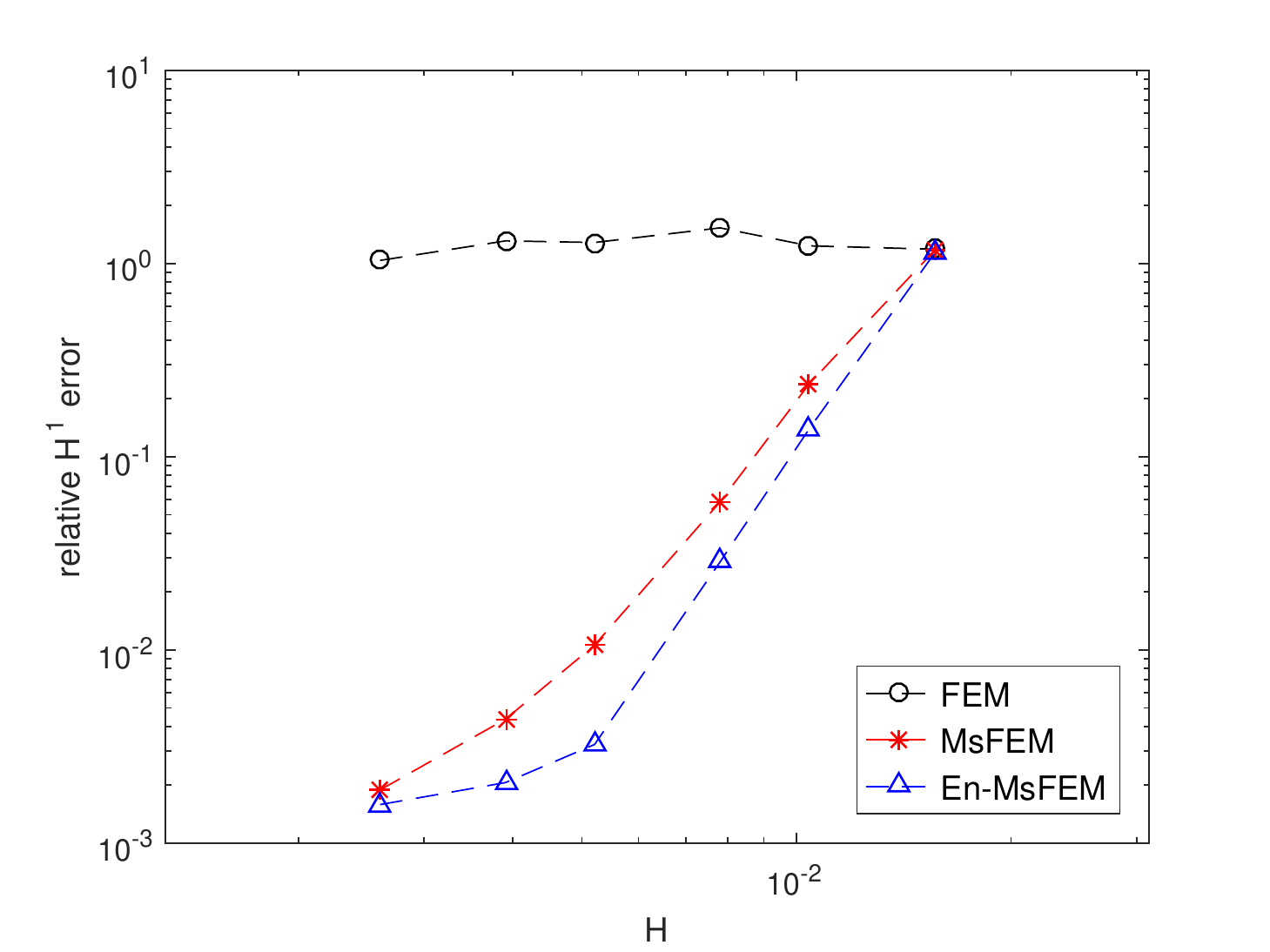}
		\caption{$H^1$ error of the wavefunction}
	\end{subfigure}
	\caption{Relative errors of wavefunction at $T=1$ in \textbf{Example \ref{example3}}. }
	\label{fig:ex3_waveerr_final_v3}
\end{figure}

\begin{figure}[htbp]
	\centering
	\begin{subfigure}{0.45\textwidth}
		\includegraphics[width=1.0\textwidth]{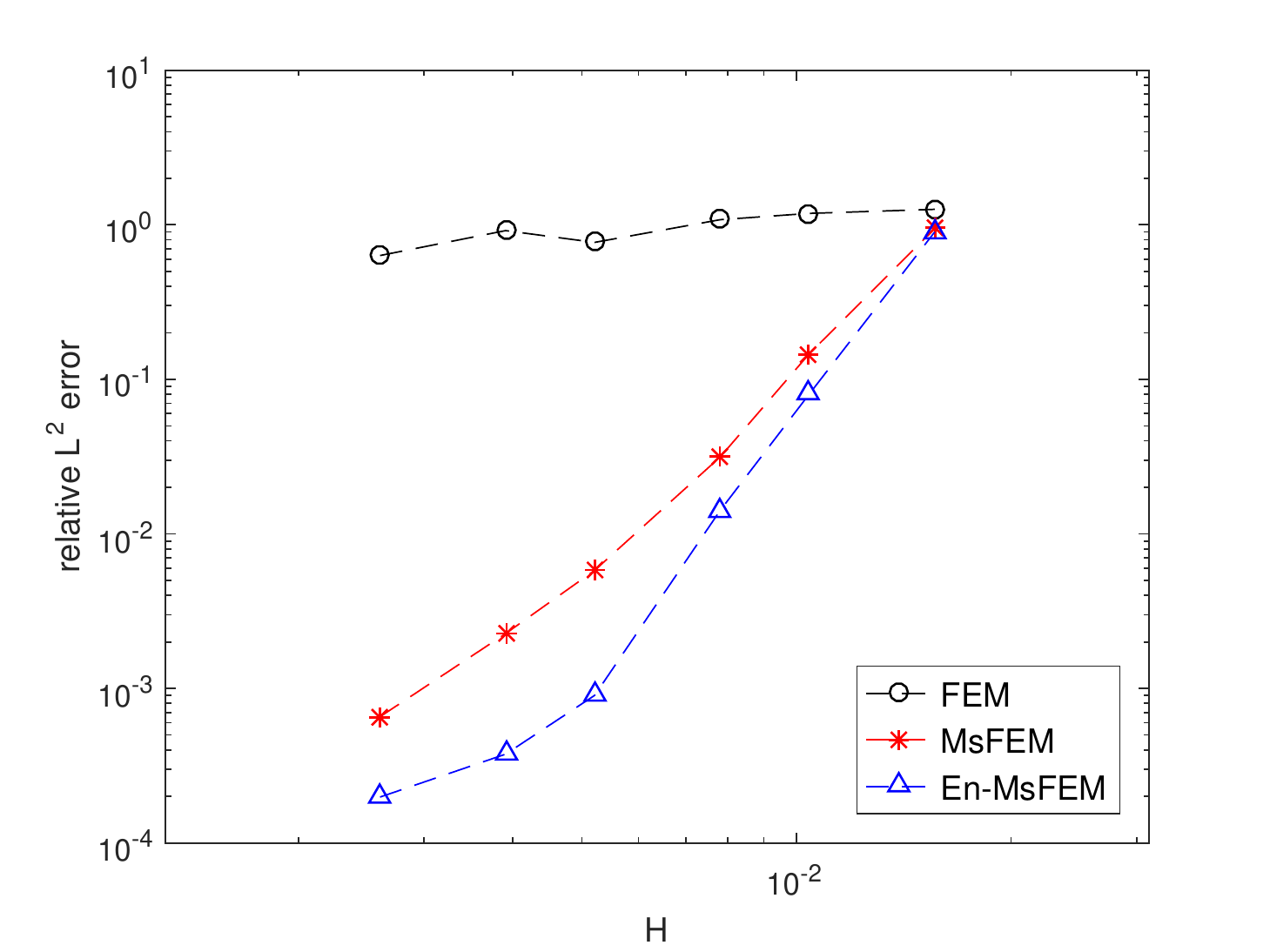}
		\caption{$L^2$ error of the position density function}
	\end{subfigure}
	\begin{subfigure}{0.45\textwidth}
		\includegraphics[width=1.0\textwidth]{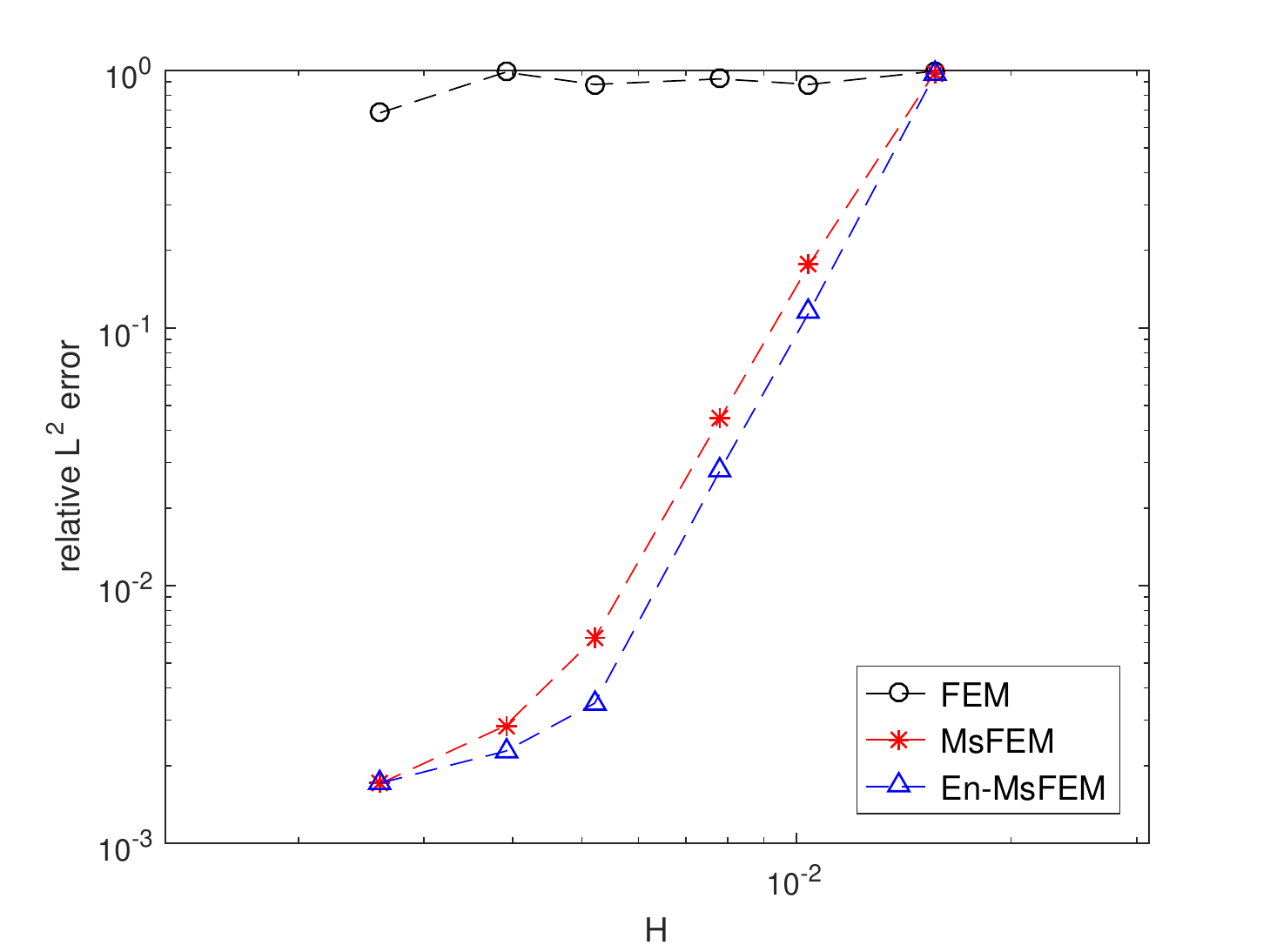}
		\caption{$L^2$ error of the energy density function}
	\end{subfigure}
	\caption{Relative errors of density functions at $T=1$ in \textbf{Example \ref{example3}}.}
	\label{fig:ex3_densityerr_final_v3}
\end{figure}

\begin{figure}[htbp]
	\centering
	\begin{subfigure}{0.45\textwidth}
		\includegraphics[width=1.0\textwidth]{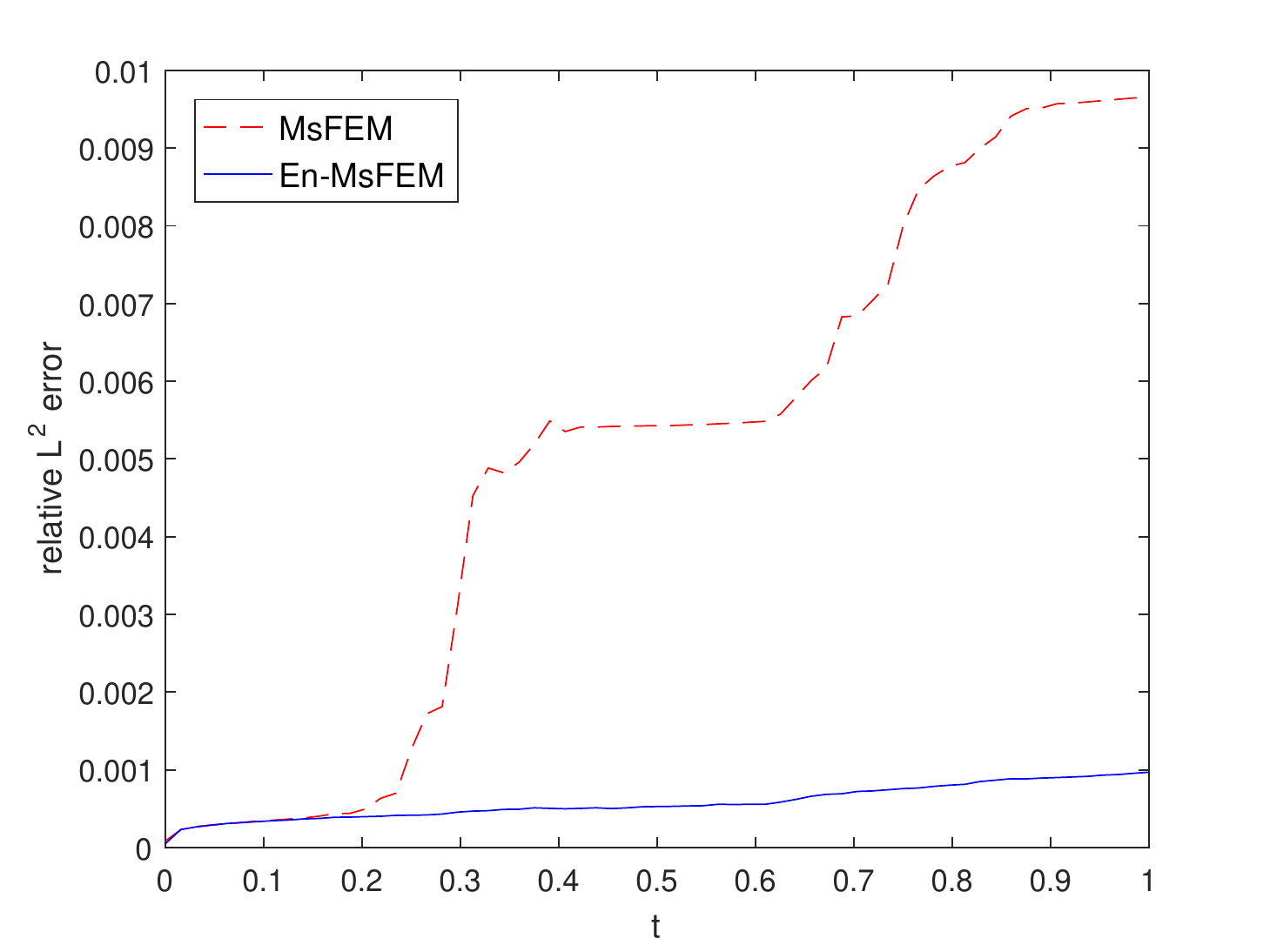}
		\caption{$L^2$ error of the wavefunction}
	\end{subfigure}
	\begin{subfigure}{0.45\textwidth}
		\centering
		\includegraphics[width=\textwidth]{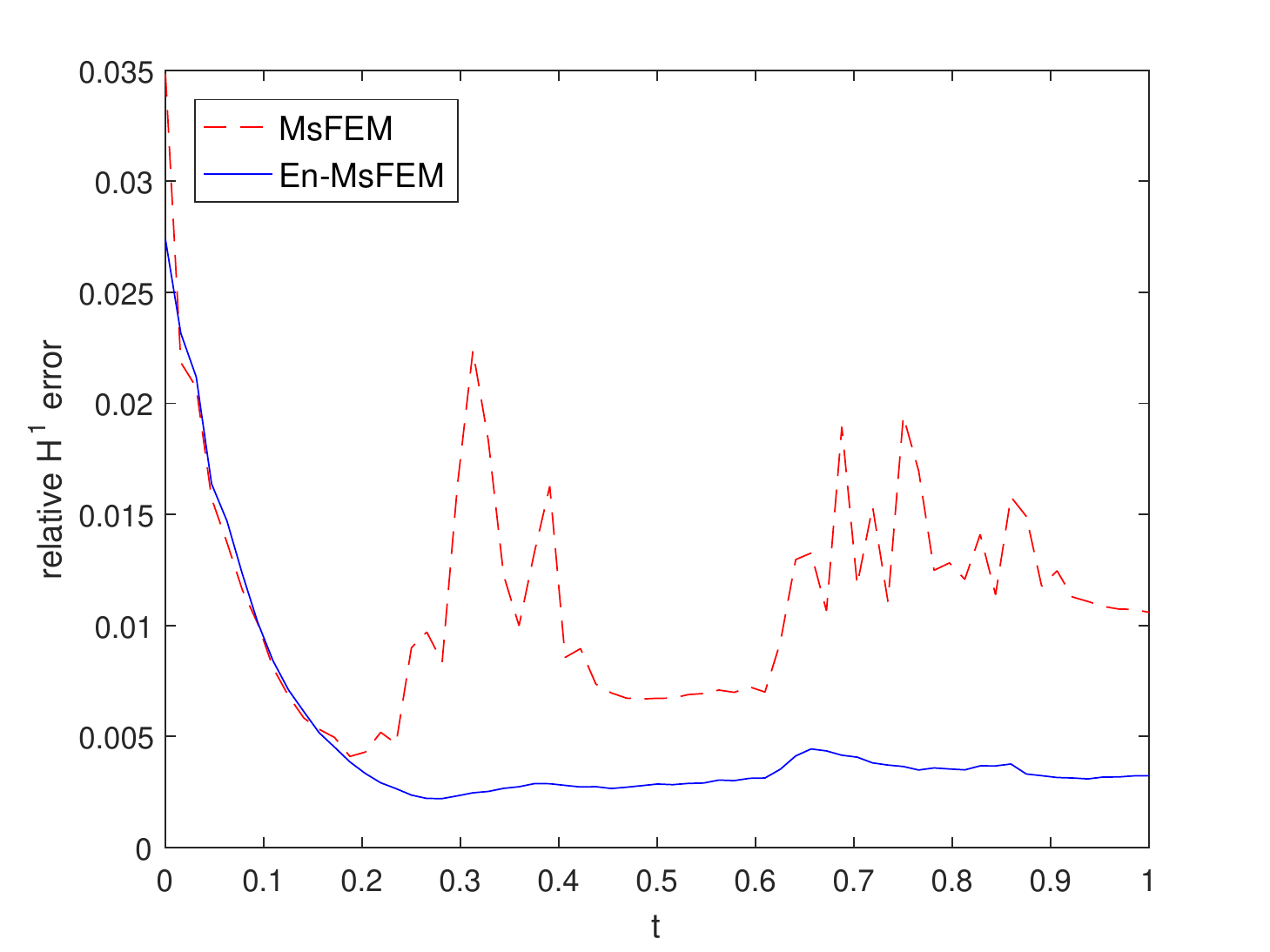}
		\caption{$H^1$ error of the wavefunction}
	\end{subfigure}
	\begin{subfigure}{0.45\textwidth}
		\includegraphics[width=1.0\textwidth]{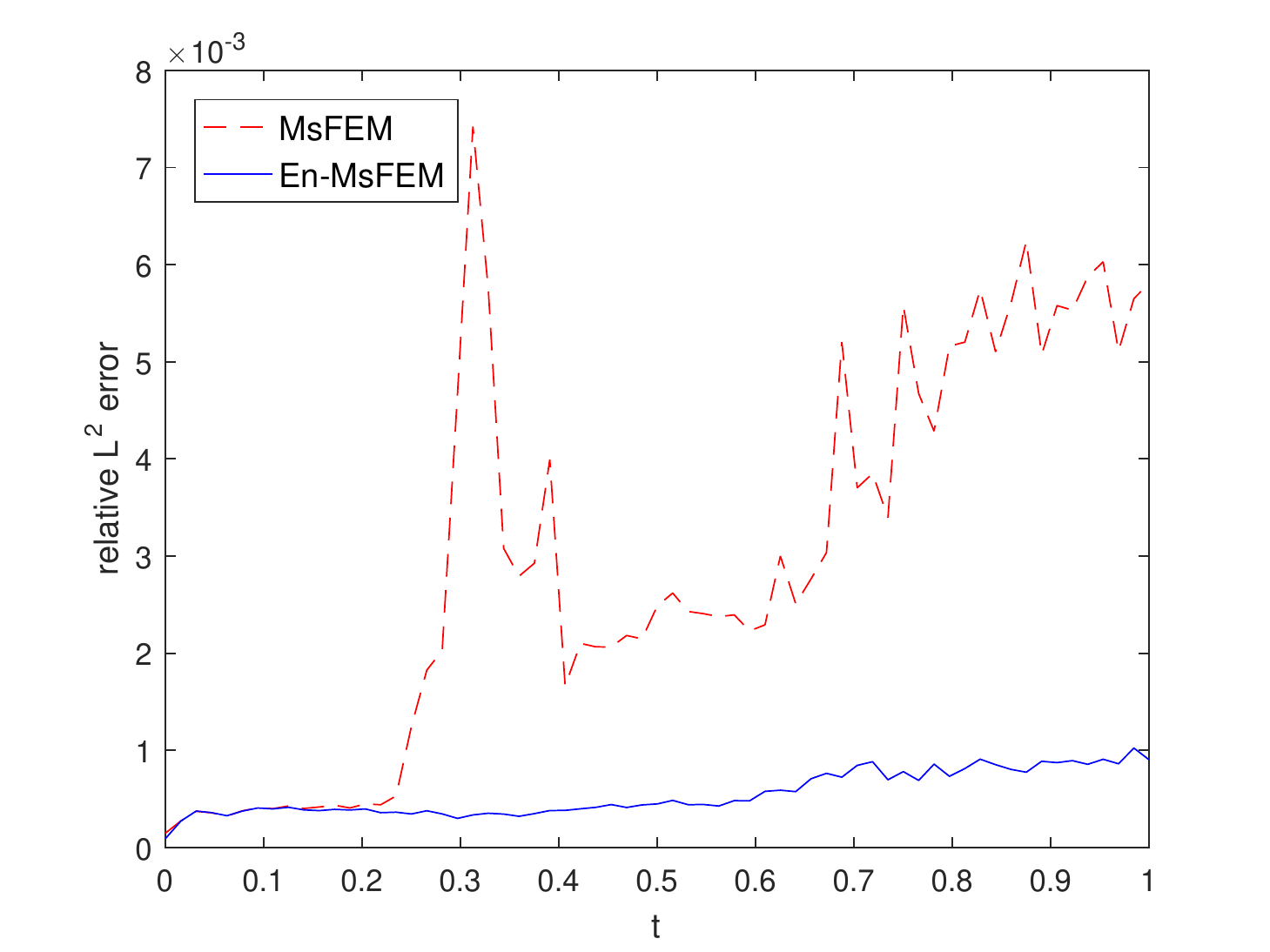}
		\caption{$L^2$ error of the position density function}
	\end{subfigure}
	\begin{subfigure}{0.45\textwidth}
		\includegraphics[width=1.0\textwidth]{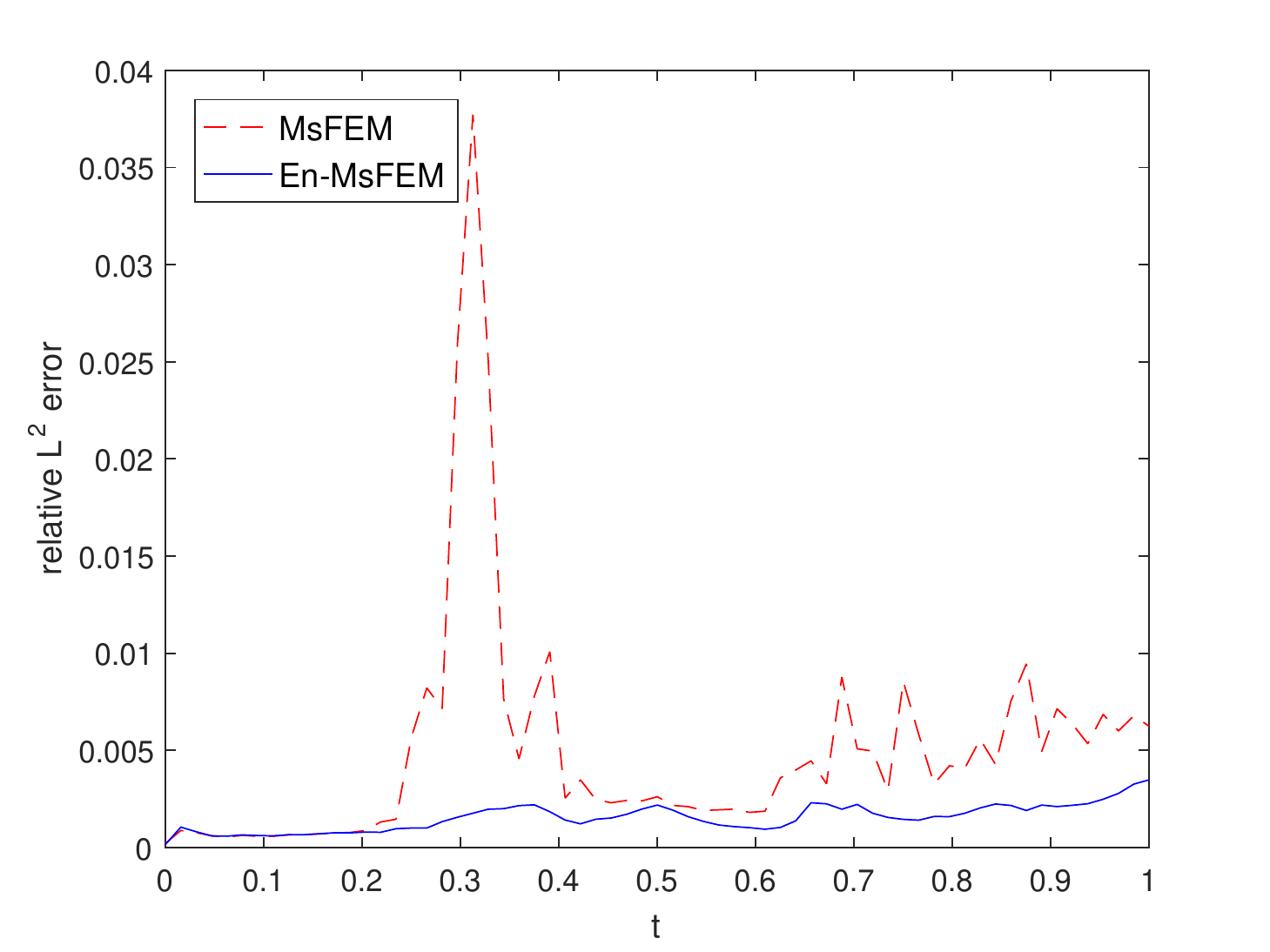}
		\caption{$L^2$ error of the energy density function}
	\end{subfigure}
	\caption{Relative $L^2$ errors of MsFEM and En-MsFEM as a function of time when $H=\frac{1}{192}$ in \textbf{Example \ref{example3}}.}
	\label{fig:ex3_err_series}
\end{figure}		
\end{example}

\begin{example}[2D case with a checkboard potential]\label{example4}
The potential $v^\epsilon(x,y,t) = v_1^\epsilon(x,y) + v_2(x,y,t)$.	The time-independent part $v_1^\epsilon(x,y)$ is a checkboard potential, which is of the following form
\begin{equation}\label{eqn:potential5}
v_1^{\epsilon}(x,y)= \left\{
\begin{aligned}
(\sin(2\pi\frac{x}{\epsilon_2})+1)(\cos(2\pi\frac{y}{\epsilon_2})),\qquad&\{0\leq x, y\leq\frac{1}{2}\}\cup\{\frac{1}{2}\leq x, y\leq 1\},\\
(\sin(2\pi\frac{x}{\epsilon}))(\cos(2\pi\frac{y}{\epsilon})+1),\qquad& \text{otherwise},\\
\end{aligned}
\right.
\end{equation}
where $\epsilon=1/8$, $\epsilon_2=1/6$. The profile of \eqref{eqn:potential5} is visualized in Figure \ref{fig:potential}, which allows
for multiple spatial scales and discontinuities around interfaces, as in quantum metamaterials \cite{Quach:2011}. The time-dependent part
is $v_2(x,y,t)=E_0\sin(2\pi t)(x+y)$ with $E_0=20$. The reference solution is obtained by En-MsFEM with $H=\frac{1}{64}$.
\begin{figure}[htbp]
	\centering
	\includegraphics[width=0.50\linewidth]{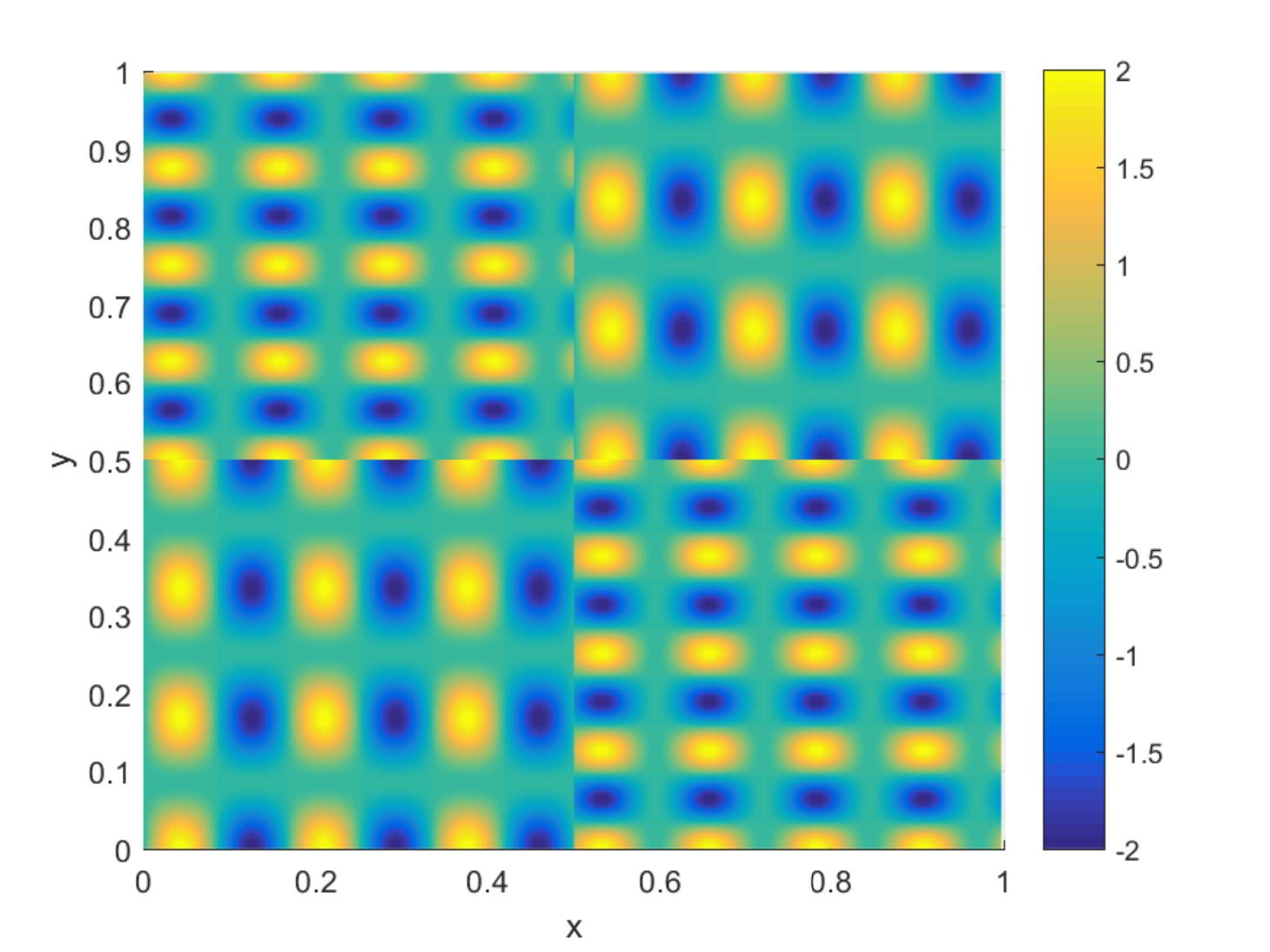}
	\caption{A checkboard-type potential over the unit square in \textbf{Example \ref{example4}}.}
	\label{fig:potential}
\end{figure}
	
Figure \ref{fig:ex4_err_final} records the relative errors in both $L^2$ norm and $H^1$ norm for a series of coarse meshes $ H = \frac{1}{16}, \frac{1}{24}, \frac{1}{32}, \frac{1}{48}$. The number of enriched basis is $\frac{1}{16}$ of that in the MsFEM, obtained at the time when $v_2(x,y,t)$ is maximized. We choose $\Delta t = \frac{1}{2^{18}}$ so the approximation error due to the temporal discretization can be ignored. 

In Figure \ref{fig:ex4_err_series}, we plot relative $L^2$ and $H^1$ errors of the wavefunction.
From these results, for moderate coarse meshes, we can see that the MsFEM reduces the approximation error by more than two orders of magnitude than that of the standard FEM in both $L^2$ and $H^1$ norms. In addition, En-MsFEM further reduces the error by about one order of magnitude in $L^2$ norm and by several times in $H^1$ norm. 

We visualize profiles of position density and energy density functions of MsFEM, En-MsFEM, and the 
standard FEM in Figure \ref{fig9} and Figure \ref{fig10}. Nice agreement is observed. 
Thus, the MsFEM and En-MsFEM provide accurate results for this 2D example.

\begin{figure}[htbp]
	\centering
	\begin{subfigure}{0.45\textwidth}
		\includegraphics[width=1.0\textwidth]{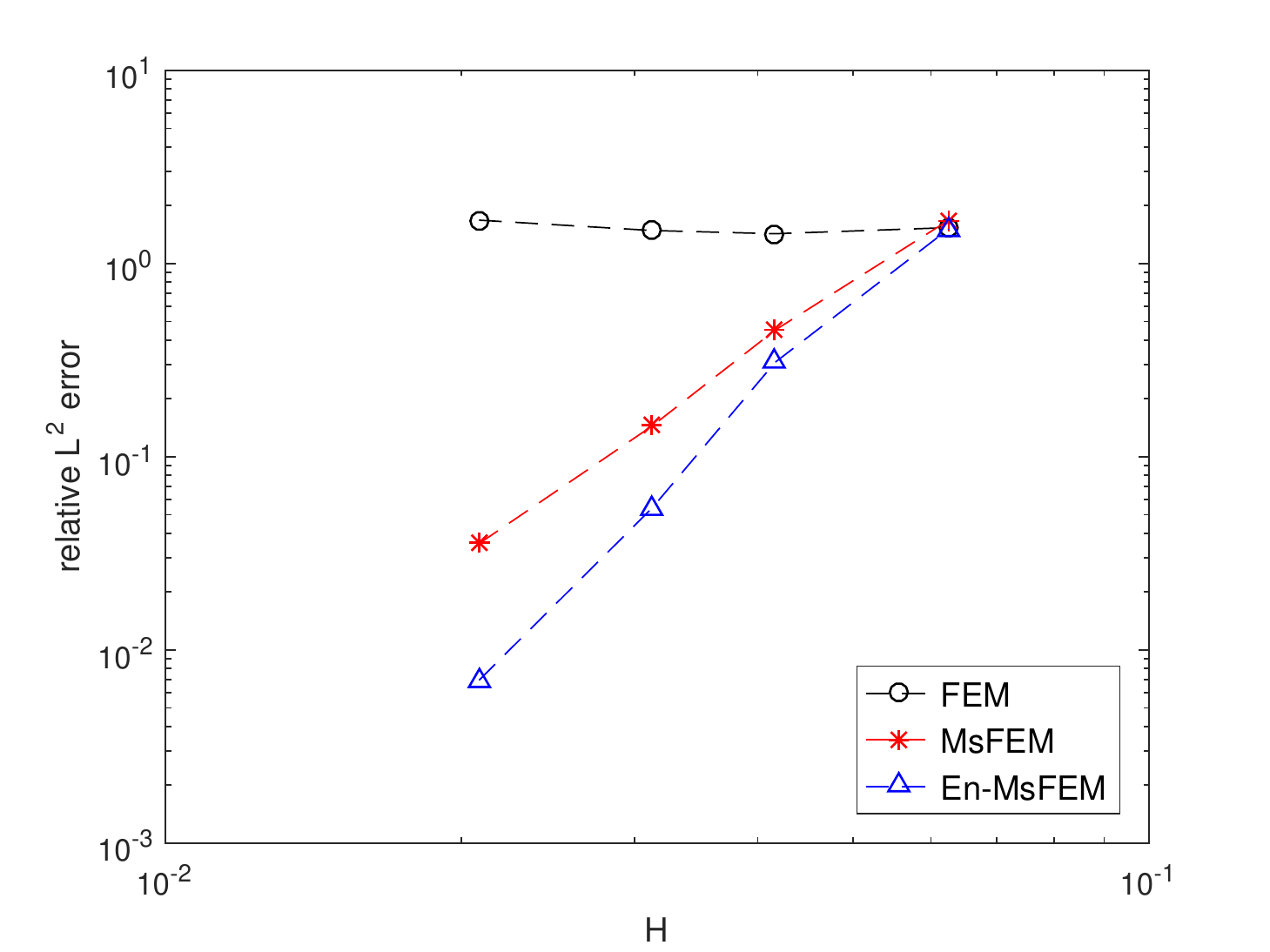}
		\caption{$L^2$ error of the wavefunction}
	\end{subfigure}
	\begin{subfigure}{0.45\textwidth}
		\centering
		\includegraphics[width=\textwidth]{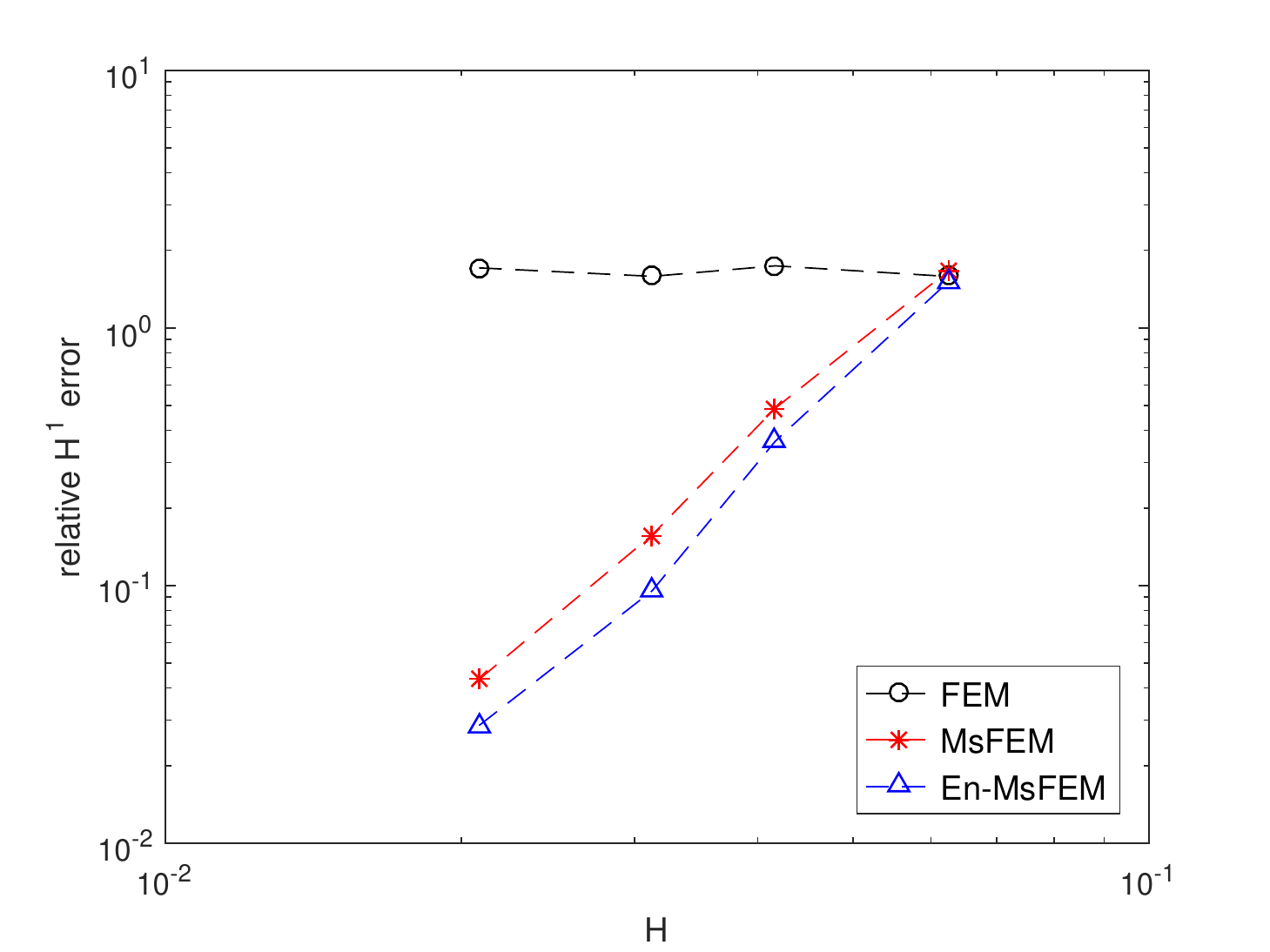}
		\caption{$H^1$ error of the wavefunction}
	\end{subfigure}
	\caption{Relative errors of wavefunction at $T=1$ in \textbf{Example \ref{example4}}.}
	\label{fig:ex4_err_final}
\end{figure}
\begin{figure}[htbp]
	\centering
	\begin{subfigure}{0.45\textwidth}
		\includegraphics[width=1.0\textwidth]{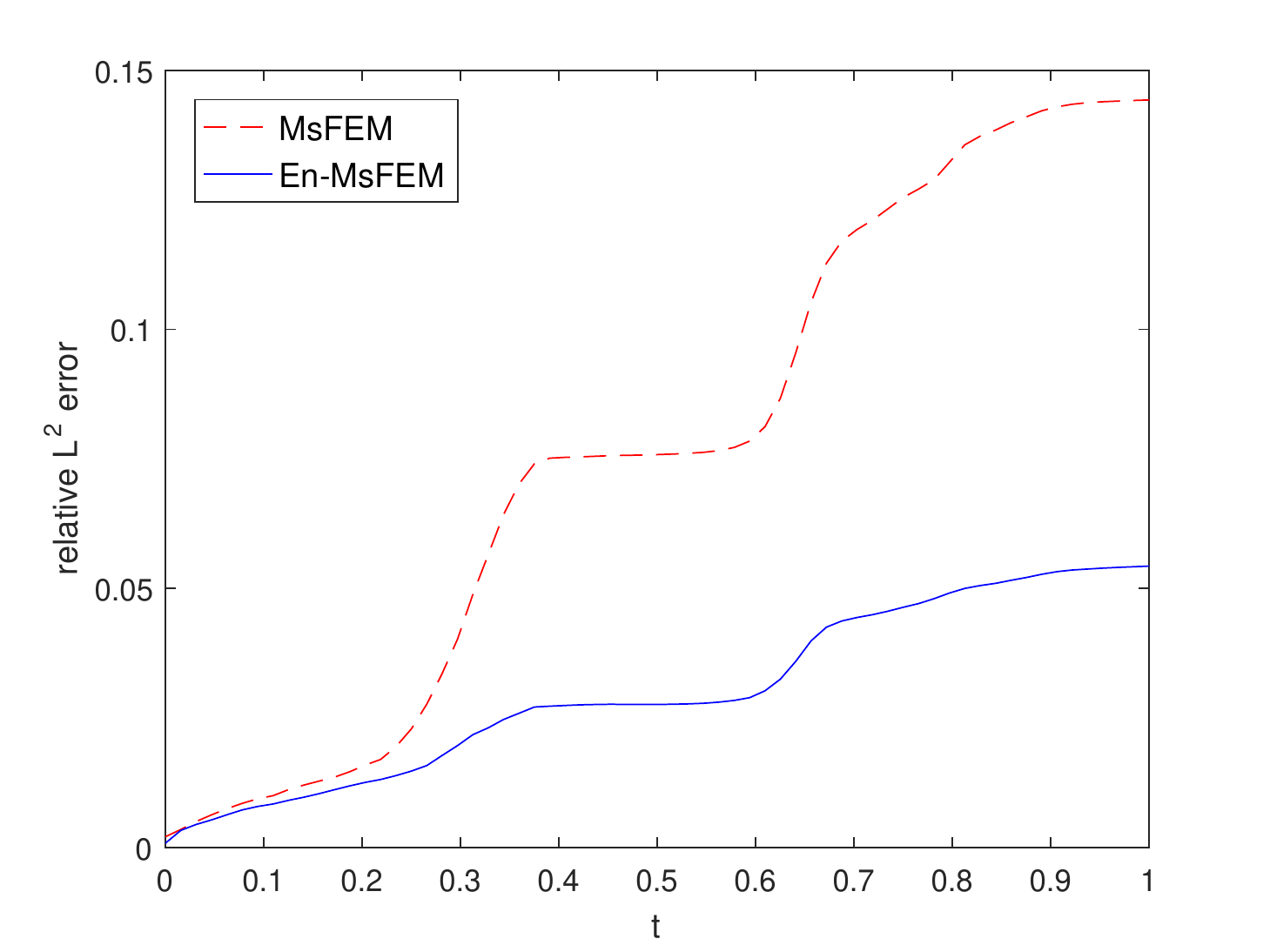}
		\caption{$L^2$ error of the wavefunction}
	\end{subfigure}
	\begin{subfigure}{0.45\textwidth}
		\centering
		\includegraphics[width=\textwidth]{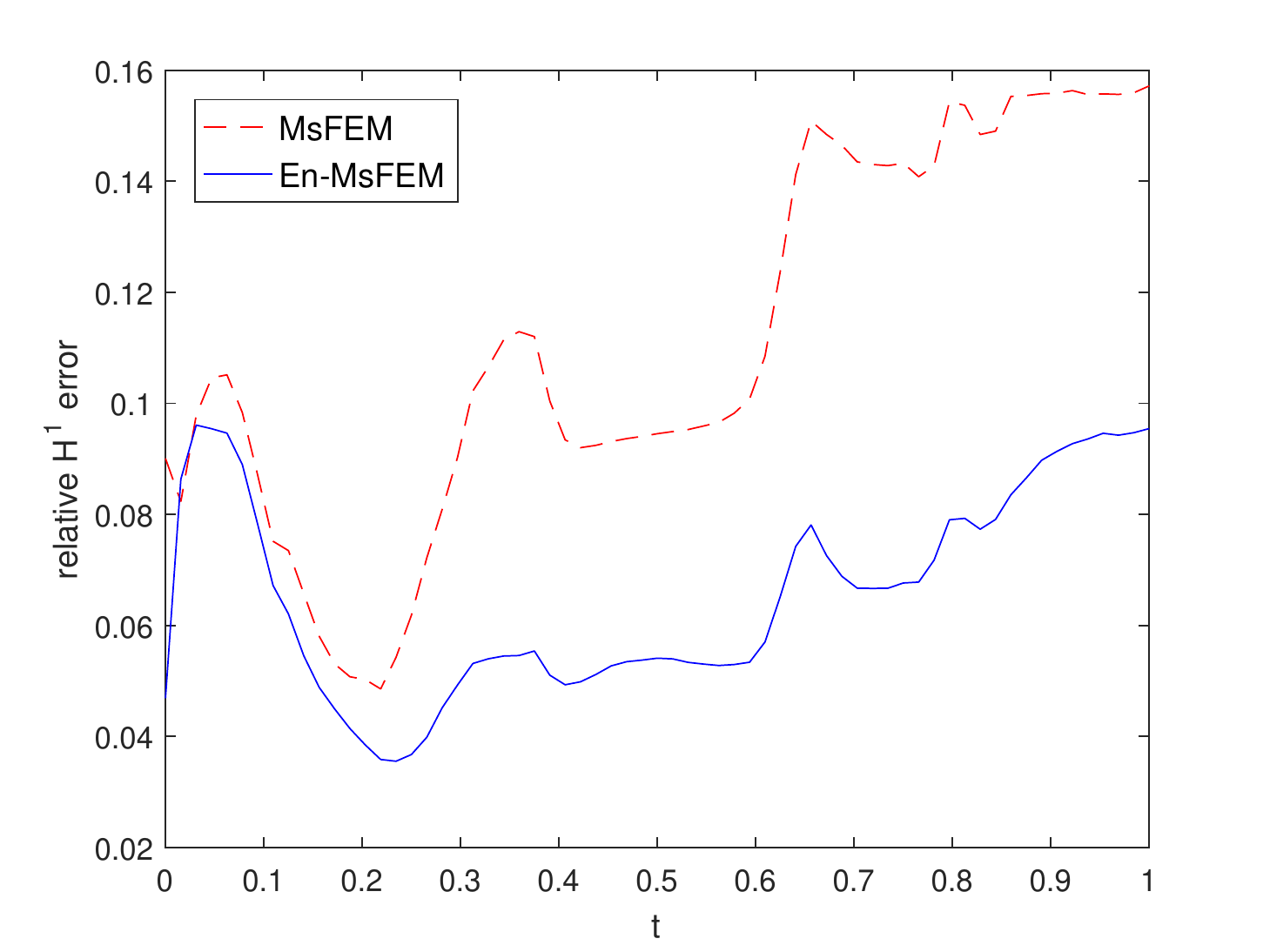}
		\caption{$H^1$ error of the wavefunction}
	\end{subfigure}
	\caption{ Relative errors of MsFEM and En-MsFEM as a function of time when $H=\frac{1}{32}$ in \textbf{Example \ref{example4}}.}
	\label{fig:ex4_err_series}
\end{figure}
	
\begin{figure}[htbp]
	\centering
	\includegraphics[width=0.92\textwidth]{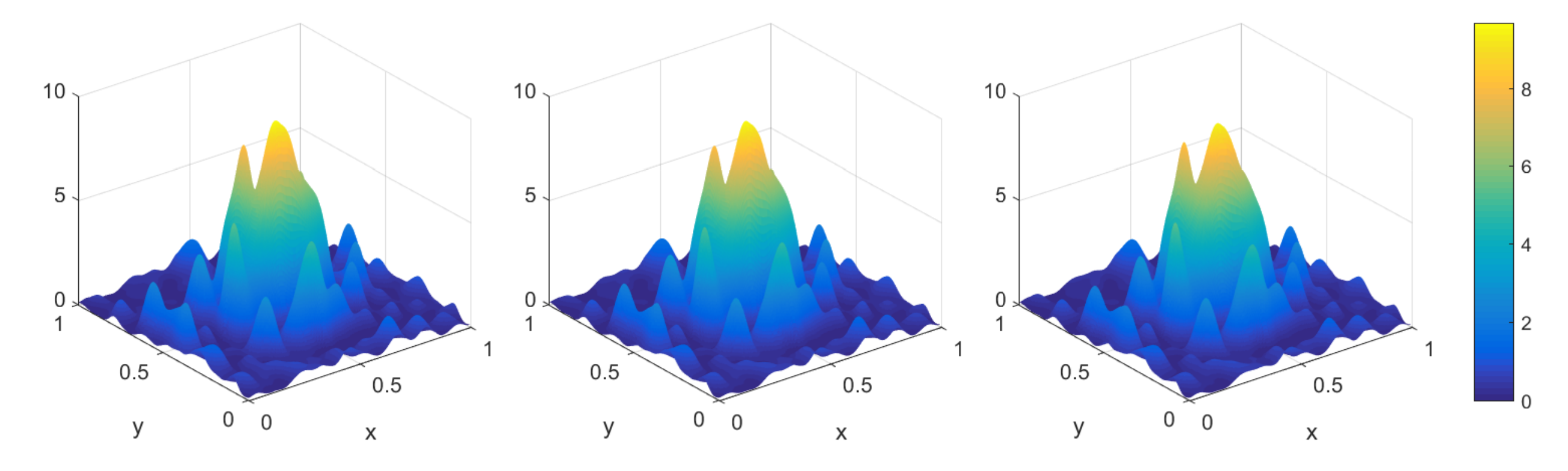}\\
	\includegraphics[width=0.92\textwidth]{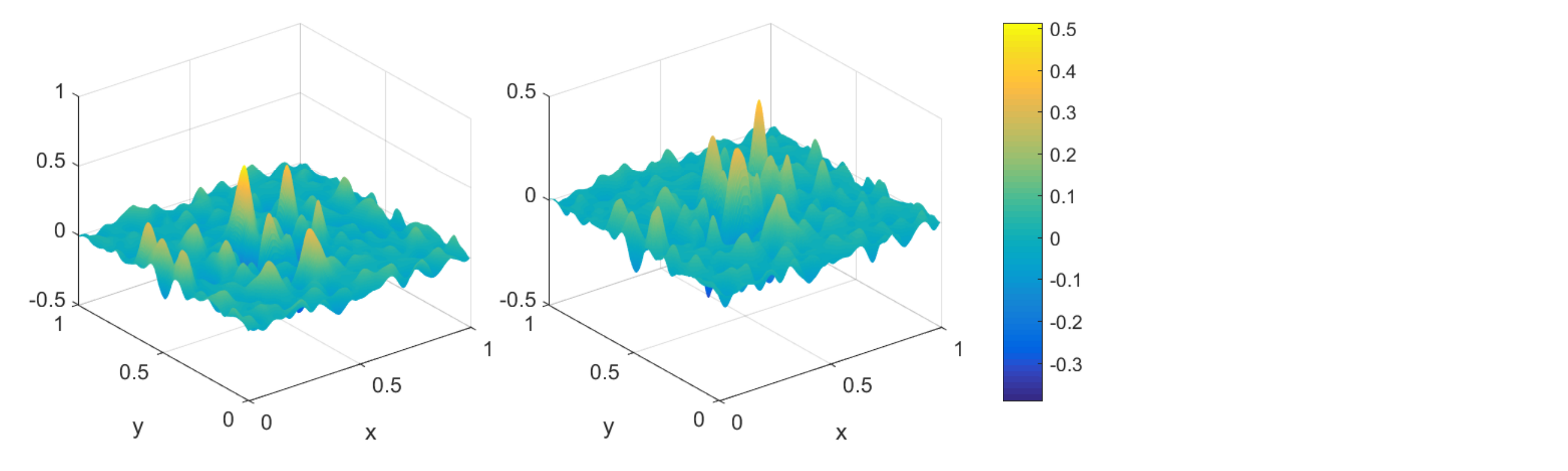}
	\caption{Profiles of position density functions at $T=1$ in \textbf{Example \ref{example4}} when $ H=\frac{1}{32}$. From left to right: MsFEM, En-MsFEM, and the reference solution with the same colorbar. Bottom row: $n_{\textrm{num}}^{\epsilon}(\bx,T)-n_{\textrm{ref}}^{\epsilon}(\bx,T)$.}
	\label{fig9}
\end{figure}
\begin{figure}[htbp]
	\centering
	\includegraphics[width=0.92\textwidth]{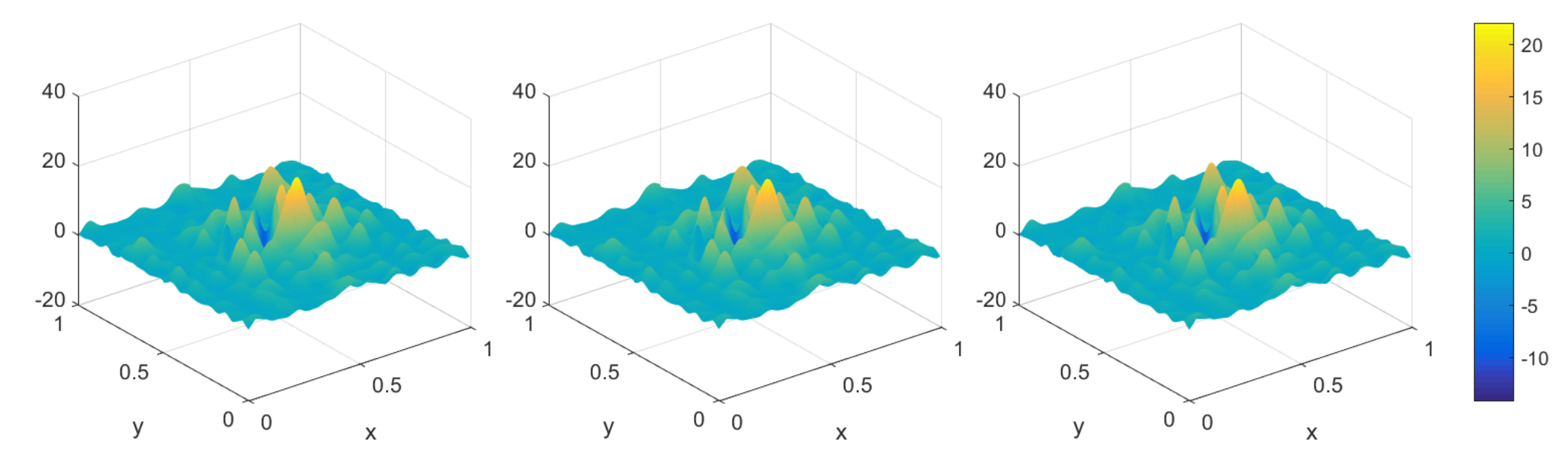}\\
	\includegraphics[width=0.92\textwidth]{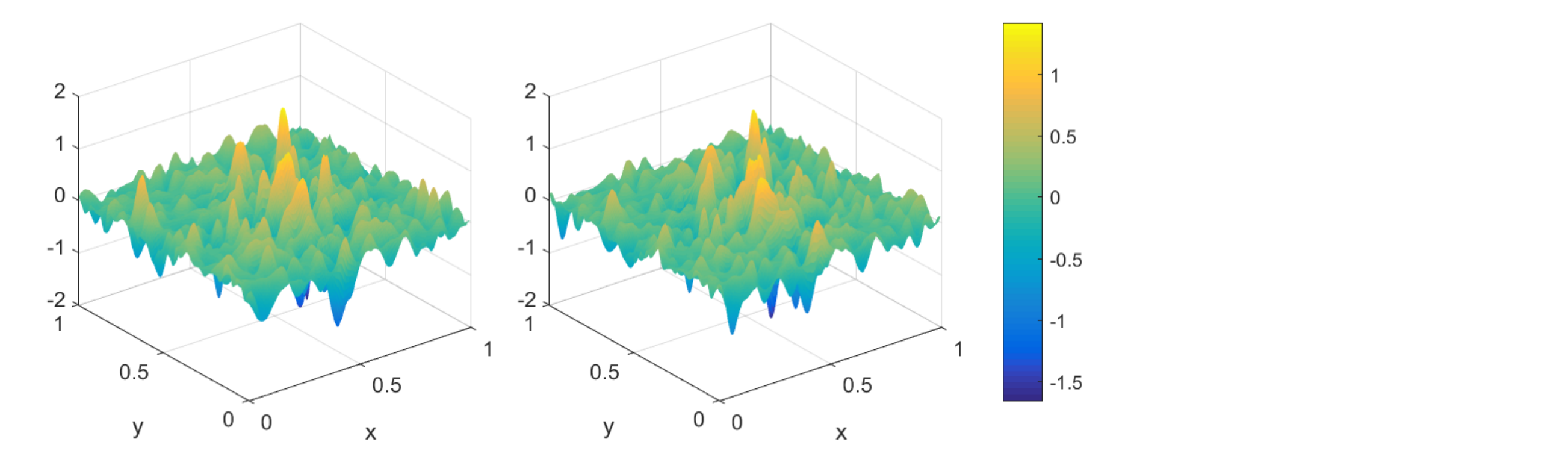}
	\caption{Profiles of energy density functions at $T=1$ in \textbf{Example \ref{example4}} when $ H=\frac{1}{32}$. From left to right in top row: MsFEM, En-MsFEM, and the reference solution with the same colorbar. Bottom row: $e_{\textrm{num}}^{\epsilon}(\bx,T)-e_{\textrm{ref}}^{\epsilon}(\bx,T)$.}
	\label{fig10}
\end{figure}
\end{example}
 	
\section{Conclusions} \label{sec:Conclusion}
\noindent	
In this paper, we have proposed two multiscale finite element methods to solve the semiclassical Schr\"{o}dinger equation with time-dependent potentials. 
In the first approach, the localized multiscale basis functions are constructed using sparse compression of the Hamiltonian operator at the initial time; in the second approach, basis functions are further enriched using a greedy algorithm for the sparse compression of the Hamiltonian operator at later times.
In the online stage, the Schr\"{o}dinger equation is approximated by these localized multiscale basis in space and is solved by Crank-Nicolson
method in time. The spatial mesh size in multiscale finite element methods is $ H=\mathcal{O}(\epsilon) $, while $H=\mathcal{O}(\epsilon^{3/2})$ in the standard finite element method. A number of numerical examples in 1D and 2D are given to demonstrate the efficiency and robustness of the proposed method.

From the perspective of physics, the proposed methods can be combined with numerical methods for Landau-Lifshitz equation \cite{ChenLiuZhou:2016} to study current-driven domain wall dynamics \cite{Chen:2015}, which are of great interest in spintronic devices and will be explored later. 

\section*{Acknowledgements}
\noindent
J. Chen acknowledges the financial support by National Natural Science Foundation of China via grants 21602149 and 11971021. The research of S. Li is partially supported by the Doris Chen Postgraduate Scholarship. 
Z. Zhang acknowledges the financial support of Hong Kong RGC grants (Projects 27300616, 17300817, and 17300318) and 
National Natural Science Foundation of China via grant 11601457, Seed Funding Programme for Basic Research (HKU), and Basic Research Programme (JCYJ20180307151603959) of The Science, Technology and Innovation Commission of Shenzhen Municipality. Part of the work was done when J. Chen was visiting Department of Mathematics, University of Hong Kong. J. Chen would like to thank its hospitality. The computations were performed using the HKU ITS research computing facilities that are supported in part by the Hong Kong UGC Special Equipment Grant (SEG HKU09).  

\appendix
\section{Continuous dependence of multiscale basis functions on the potential function}\label{PhiDependOnV} 
\noindent
In this appendix, we prove \textbf{Theorem \ref{basis-dependon-potential}}, which plays an important role 
in the enrichment of multiscale basis functions. 
\begin{proof}
	For each time instance $t_{\ell}$, when numerically solving \eqref{OC_SchBasis_Obj} - \eqref{OC_SchBasis_Cons2}, we have the 
	following quadratic programming problem with equality constraints
	\begin{equation}
	\left\{
	\begin{aligned}
	& \min_{c} \frac{1}{2}c^T Q c,\\
	& \textrm{subject to } Ac=b,
	\end{aligned}
	\right.
	\label{eqn:QP}
	\end{equation}
	where $Q$ is a symmetric positive definite matrix on the fine triangularization $\mathcal{T}_h$ with the $(i,j)$ component
	\[
	Q_{ij}=\frac{\veps^2}{2}(\nabla\varphi^h_{j},\nabla\varphi^h_{i})+(v_1^{\veps}(\bx)\varphi^h_j,\varphi^h_{i})
	+(v_2(\bx,t_{\ell})\varphi^h_j,\varphi^h_{i}),
	\]
	and $A$ is a long matrix with $b$ a long vector coming from \eqref{OC_SchBasis_Cons1} - \eqref{OC_SchBasis_Cons2}.
	
	Under the assumptions that $v_1^{\epsilon}(\bx)+v_2(\bx,t_{\ell})$ is uniformly bounded and $h/\veps=\kappa $ is small,
	we know that $Q$ is a positive definite matrix. Moreover, we know that $A$ has full rank, i.e., $\rank(A)=N_H$. Therefore, the quadratic optimization problem \eqref{eqn:QP} has a unique minimizer, satisfying the Karush-Kuhn-Tucker condition. Specifically, the unique minimizer of \eqref{eqn:QP} can be explicitly written as
	\begin{align}\label{eqn:minref}
	\boldsymbol{c} = Q^{-1} A^T (AQ^{-1} A^T)^{-1}\boldsymbol{b}.
	\end{align}
	For two time instances $t_{\ell_1}$ and $t_{\ell_2}$, we define $\delta V = Q_1- Q_2$. Then
	\begin{align}
	\left(\delta V\right)_{ij}= \left(\left(v_2(\cdot,t_{\ell_1})-v_2(\cdot,t_{\ell_2})\right)\varphi^h_i,\varphi^h_j\right),
	\end{align}
	and thus
	\begin{align}
	\lVert \delta V \rVert_{\infty} \leq h^d  \lVert v_2(\cdot,t_{\ell_1}) - v_2(\cdot,t_{\ell_2}) \rVert_{L^{\infty}(D)}.
	\end{align}
	We choose $h$ to be small enough such that $\lVert \delta V \rVert_{\infty} \leq 1$, and have
	\[
	Q_2^{-1} = \sum_{n=0}^{\infty} \left(Q_1^{-1}\delta V\right)^n Q_1^{-1},
	\]
	and thus
	\begin{align*}
	\boldsymbol{c}_2 - \boldsymbol{c}_1 & = \left[ Q^{-1}_2 - Q^{-1}_1 \right] A^T (AQ^{-1}_1 A^T)^{-1}\boldsymbol{b} 
	+ Q^{-1}_2 A^T \left[ (AQ^{-1}_2 A^T)^{-1} - (AQ^{-1}_1 A^T)^{-1}\right]\boldsymbol{b}, \nonumber\\
	& = Q^{-1}_1 \delta V Q^{-1}_1  A^T (AQ^{-1}_1 A^T)^{-1}\boldsymbol{b}  \nonumber \\
	& \quad - Q^{-1}_2 A^T (AQ^{-1}_1 A^T)^{-1} (AQ^{-1}_1 \delta V Q^{-1}_1 A^T) (AQ^{-1}_1 A^T)^{-1}\boldsymbol{b} + o(\lVert\delta V\rVert_{\infty}),\nonumber \\
	& = Q^{-1}_1 \delta V Q^{-1}_1  A^T (AQ^{-1}_1 A^T)^{-1}\boldsymbol{b} \nonumber \\
	& \quad - Q^{-1}_1 A^T (AQ^{-1}_1 A^T)^{-1} (AQ^{-1}_1 \delta V Q^{-1}_1 A^T) (AQ^{-1}_1 A^T)^{-1}\boldsymbol{b} + o(\lVert\delta V\rVert_{\infty}).\nonumber 
	\end{align*}
	Therefore,
	\begin{align*}
	\lvert \boldsymbol{c}_2 - \boldsymbol{c}_1 \rvert_{\infty} & \leq 
	C \lVert A \rVert_{\infty} \lVert Q^{-1}_1 \rVert_{\infty}^2 
	\lVert (AQ^{-1}_1A^T)^{-1} \rVert_{\infty}
	\lvert \boldsymbol{b} \rvert_{\infty} \left( 1 + 
	\lVert A \rVert_{\infty}^2 \lVert Q^{-1}_1 \rVert_{\infty} \lVert (AQ^{-1}_1A^T)^{-1} \rVert_{\infty} \right) \lVert \delta V \rVert_{\infty} .\nonumber 
	\end{align*}
	By their definitions, we have
	\[
	\lVert A \rVert_{\infty} \leq C h^d , \quad \lvert \boldsymbol{b} \rvert_{\infty} =1, \quad \lVert Q^{-1}_1 \rVert_{\infty} \leq C h^{-2}, \quad \lVert Q_1 \rVert_{\infty} \leq C \max\{\veps^2,h^2\}\leq C\veps^2,
	\]
	and thus
	\begin{align*}
	\lvert \boldsymbol{c}_2 - \boldsymbol{c}_1 \rvert_{\infty} & \leq 
	C \veps^{4} h^{-6} h^{-d} \lVert \delta V \rVert_{\infty}
	\leq C \veps^{4} h^{-6} \lVert v_2(\cdot,t_{\ell_2}) - v_2(\cdot,t_{\ell_1}) \rVert_{L^{\infty}(D)}.\nonumber 
	\end{align*}
	We complete the proof since $h/\veps=\kappa $ and $\lVert \phi(\cdot,t_{\ell_2}) - \phi(\cdot,t_{\ell_1}) \rVert_{L^{\infty}(D)} \leq \lvert \boldsymbol{c}_2 - \boldsymbol{c}_1 \rvert_{\infty}$.
\end{proof}		
\bibliographystyle{siam}
\bibliography{ZWpaper}
		
\end{document}